\documentclass[twoside,leqno,twocolumn]{article}

\usepackage[letterpaper]{geometry}

\usepackage{ltexpprt}

\usepackage{amsmath}
\usepackage{amssymb}
\usepackage{numprint}
\usepackage{hyperref}
\usepackage{graphicx}
\usepackage[title,toc,titletoc,page]{appendix}
\usepackage{flushend}
\usepackage{dblfloatfix}
\usepackage{xcolor}
\usepackage{wrapfig}
\usepackage{graphicx}
\usepackage{subcaption}
\usepackage{nccmath}
\usepackage{fancyhdr}
\usepackage{float}

\npdecimalsign{.}

\usepackage{caption}

\usepackage{algorithm}
\usepackage{algpseudocode}

\newtheorem{claim}{Claim}

\newtheorem{innercustomgeneric}{\customgenericname}
\providecommand{\customgenericname}{}
\newcommand{\newcustomtheorem}[2]{%
  \newenvironment{#1}[1]
  {%
   \renewcommand\customgenericname{#2}%
   \renewcommand\theinnercustomgeneric{##1}%
   \innercustomgeneric
  }
  {\endinnercustomgeneric}
}

\newcustomtheorem{customclaim}{Claim}
\newcustomtheorem{customlemma}{Lemma}

\algnewcommand\algorithmicinput{\textbf{Input:}}
\algnewcommand\INPUT{\item[\algorithmicinput]}
\algnewcommand\algorithmicoutput{\textbf{Output:}}
\algnewcommand\OUTPUT{\item[\algorithmicoutput]}

\newcommand{\Oh}[1]{\mathcal{O}\!\left( #1\right)}

\newcommand{\junk}[1]{{}}
\newif\ifDoubleBlind
\DoubleBlindfalse

\newcommand{\ie}{i.e.\ }
\newcommand{\etal}{et~al.~}

\newcommand{\CC}{C\texttt{++}}
\newcommand{\match}{\mathcal{M}}
\newcommand{\palette}{\mathcal{C}}
\newcommand{\col}{\xi}

\newcommand{\colrec}{\texttt{RecurseCol}}
\newcommand{\colcount}{\texttt{CountCol}}
\newcommand{\colhp}{\texttt{RandRCol}}
\newcommand{\colbgkls}{\texttt{HierCol}}

\newcommand{\matchtriv}{\texttt{TrivialMatch}}
\newcommand{\matchbgs}{\texttt{Hier1Match}}
\newcommand{\matchsolomon}{\texttt{Hier2Match}}
\newcommand{\matchbbhss}{\texttt{RandR1Match}}
\newcommand{\matchrr}{\texttt{RandR2Match}}

\newcommand*{\shorterDots}{.\kern-0.06em.\kern-0.06em.} 

\newcommand{\papertitle}{Random Rank-Based, Hierarchical or Trivial: Which Dynamic Graph Algorithm Performs Best
in Practice?}
\ifDoubleBlind
\newcommand{\mytitle}{\papertitle}
\else
\newcommand{\mytitle}{\papertitle}
\fi{}

\date{}

\ifDoubleBlind
\author{Double Blind}
\else
\author{Monika Henzinger\thanks{University of Vienna, Faculty of Computer Science, Vienna, Austria, \texttt{monika.henzinger@univie.ac.at}} 
\and Alexander Noe \thanks{University of Vienna, Faculty of Computer Science, Vienna, Austria, \texttt{alexander.noe@univie.ac.at}}}
\fi{}

\begin{document}

\title{\mytitle}

\maketitle

\begin{abstract}
Fully dynamic graph algorithms that achieve polylogarithmic or better time per operation use either a hierarchical graph decomposition or random-rank based approach. 
There are so far two graph properties for which efficient algorithms for both types of data structures exist, namely fully dynamic $(\Delta + 1)$ coloring and fully dynamic maximal matching.
In this paper we present an extensive experimental study of these two types of algorithms for these two problems together with very simple baseline algorithms to determine which of these algorithms are the fastest.
Our results indicate that the data structures used by the different algorithms
dominate their performance.

\end{abstract}
\ifDoubleBlind
\vfill
\clearpage
\setcounter{page}{1}
\pagenumbering{arabic}
\fancyfoot[C]{\thepage}
\pagestyle{fancy}
\else
\fi{}

\section{Introduction} 
Real-world graphs are huge and changing dynamically, creating the need for data structures that efficiently maintain properties of such graphs. Such data structures
are called \emph{fully dynamic graph algorithms} and usually they assume that the number
of vertices is fixed and each update operation either inserts or deletes an edge.
While there exist graph properties, such as single-source shortest paths distances or maximum flow values in weighted graphs, for which (under some popular complexity assumptions) no sublinear update time is possible~\cite{abboud2014popular,henzinger2015unifying,DBLP:conf/icalp/Dahlgaard16}, there also exists other graph properties where algorithms with polylogarithmic or even constant update times are known.
Such graph properties are connectivity~\cite{DBLP:journals/jacm/HenzingerK99,HolmLT98}, minimum spanning tree~\cite{HolmLT98}, and maximal independent set~\cite{behnezhad2019fully,ChechikZ19},
which all have polylogarithmic time per operation, and
maximal matching~\cite{solomon2016fully}, $(\Delta+1)$-vertex coloring~\cite{henzinger2020constant,bhattacharya2019fully}, where $\Delta$ is the maximum degree in the graph, and $(1+\epsilon)$-approximate minimum spanning tree value~\cite{DBLP:journals/corr/abs-2011-00977}, which all have constant update time.
All (non-trivial) 
dynamic algorithms with polylogarithmic or faster update time used some variant of hierarchical decomposition of a graph into logarithmic ``layers''~\cite{DBLP:journals/jacm/HenzingerK99,HolmLT98,solomon2016fully,bhattacharya2019fully,DBLP:journals/corr/abs-2011-00977} but in 2019 a new technique, based on giving either each vertex or each edge in the graph a random value of $[0,1]$, called \emph{rank}, was introduced into the field independently in three papers~\cite{behnezhad2019fully,ChechikZ19,henzinger2020constant}.

As the asymptotic running times achieved by the two techniques are similar an obvious question is how they compare in a empirical evaluation. The purpose of this paper is to answer exactly this question. 
As $(\Delta+1)$-vertex coloring and maximal matching are the only graph properties for which efficient dynamic algorithms exist that use either one or the other technique we implemented the best dynamic algorithms for these problems, together with one or two trivial dynamic algorithms as a baseline.
Let $n$ denote the number of nodes in the graph, $m$ its current number of edges, and $\Delta$ its maximum degree over the whole sequence of operations.
Specifically, we compared the following algorithms:

(a) For maximal matching we compared the random-rank based algorithm of Behnezhad et al.~\cite{behnezhad2019fully} which takes $O(\log^4 n)$ amortized update time, the hierarchical algorithm of Solomon~\cite{solomon2016fully} which takes $O(1)$ amortized update time,
the 2-level version of the hierarchical algorithm of Baswana, Gupta and Sen~\cite{baswana2015fully} which takes $O(\sqrt n)$ amortized update time, our own variant of~\cite{behnezhad2019fully} with $O(nm)$ worst-case update time, and a trivial dynamic matching algorithm with $O(n)$ worst-case matching time.

(b) For $(\Delta+1)$-vertex coloring we compared the random-rank based algorithm of Henzinger and Peng~\cite{henzinger2020constant} that takes constant amortized time per operation with the hierarchical algorithm of Bhattacharya et al.~\cite{bhattacharya2019fully}, which also takes constant time per operation, and two simple trivial algorithms;
the first of them takes $O(\Delta)$ worst-case time per operation and the second one might not terminate at all in the worst-case. Note that all but the last algorithm use $\Theta(n \Delta)$ space, which kept us evaluating them for values of $\Delta$ close to $n$.

\textbf{Related work.}
Recently there has been a large body of work on the experimental evaluation of fully dynamic algorithms, see~\cite{DBLP:journals/corr/abs-2102-11169} for a survey. The work most closely related to our work is a recent experimental evaluation of near-optimum maximum cardinality
matching algorithms by Henzinger et al.~\cite{henzinger2020dynamic}. However, the emphasis in that work was on the quality of approximation, i.e., how close the cardinality of the matching of the different algorithms is to the \emph{maximum} cardinality matching, while in this work we evaluate algorithms for \emph{maximal} matchings and our goal is to analyze the time per update operation. For comparison we also included the implementation of the 2-level variant of the algorithm by Baswana, Gupta and Sen~\cite{baswana2015fully} of~\cite{henzinger2020dynamic} into our evaluation.

\textbf{Our results.}
We perform an extensive evaluation of four dynamic $(\Delta+1)$-coloring and five dynamic maximal matching algorithms on three different types of dynamic graph instances.
Our results show that random-rank based algorithm are superior to hierarchy-based algorithms and the trivial algorithms outperform both. The reason is the caching behavior of the different algorithms. The trivial algorithms use simple data structures as they only need to access the neighbors of a given vertex sequentially. The simple data structures can be updated significantly faster then the more sophisticated data structures needed for the other algorithms. Additionally, the access to the data structures exhibits a strong spacial locality which is well supported by system caches. The more sophisticated algorithms need to be able to look up and modify information stored for any given neighbor, which requires the use of more sophisticated and slower data structures. Additionally, their accesses within these data structures show no strong spacial or temporal locality, leading to poorer cache performance.
We also compared the hierarchy-based and the random-rank based algorithms. The latter are $10 - 50\%$ faster than the former.
The main contributing factor is that the hierarchy-based algorithms spend $~20-40\%$ of their time on maintaining the hierarchy, depending on the number of levels they use, which is not necessary for random-rank based algorithms and which explains to a large part their slower performance.

\section{Preliminaries}\label{s:preliminaries}

\subsection{Basic Concepts.}

Let $G = (V, E)$ be an unweighted and undirected graph with vertex set $V$
and edge set $E \subset V \times V$. Let $n = |V|$. The degree of a vertex
$v$ is equal to the size of the neighborhood $|N(v)|$. Here, we look at
\emph{fully dynamic} graphs, the edge set $E$ is empty initially and changes via
a sequence of \emph{edge updates}, each of them is either an \emph{insertion} or
a \emph{deletion} of an edge. The vertex set $V$ remains unchanged over the whole
update sequence. Let $\Delta$ be a parameter so that the maximum degree in $G$
remains upper bounded by $\Delta$ throughout the update sequence.

A \emph{matching} $\match$ in a graph $G$ is a subset of the edge set $\match
\subseteq E$ so that each vertex $v$ has only up to one incident edge $e \in
\match$. A vertex that has an incident edge $\in \match$ is \emph{matched},
otherwise a vertex is \emph{unmatched}. A matching is \emph{maximal}, if no edge
$e = (u,v) \in E$ with $e \not \in \match$ can be added to the matching, as
either $u$ or $v$ are already matched. For a edge $e = (u,v) \in \match$, $u$ and
$v$ are the \emph{partner} $P$ of the respective other vertex, $P(u) = v$ and
$P(v) = u$. For a vertex $v \in V$ without an incident matching edge, $P(v) =
\bot$.

For an undirected graph $G = (V,E)$, an integral parameter $\lambda > 0$ and a
\emph{color palette} $\palette = \{1, \dots, \lambda\}$, let a
$\lambda$-coloring in $G$ be a function $\col: V \rightarrow \palette$ which
assigns a color $\col(v) \in \palette$ to each vertex $v \in V$. A coloring is
\emph{proper} if no two neighboring vertices in $G$ have the same color in
$\col$.

\section[Algorithms for the Fully Dynamic Coloring Problem]{Algorithms for the Fully Dynamic $(\Delta + 1)$ Coloring Problem}
\label{sec:alg-color}

In this Section we briefly outline the implemented algorithms for the
\emph{fully dynamic $(\Delta + 1)$ coloring problem}. We implement two simple
algorithms as well as the recent constant time algorithms of Henzinger and
Peng~\cite{henzinger2020constant} and of Bhattacharya et
al.~\cite{bhattacharya2019fully}. The algorithm of Henzinger and Peng uses
random vertex ranks while the algorithm of Bhattacharya et al. is based on a
hierarchical partition of the vertex set. We briefly outline the central ideas
of the algorithms, starting with the simple ones, for further details of the 
often
lengthy algorithm descriptions in~\cite{henzinger2020constant,bhattacharya2019fully,baswana2015fully,solomon2016fully} we refer the reader to the original works. Note that the algorithms are all randomized.

\subsection[Recursive Fully-Dynamic Coloring]{Recursive Fully-Dynamic $(\Delta + 1)$ Coloring.}

For a fully dynamic graph with maximum degree  $\Delta$, we can give a trivial
$(\Delta + 1)$-coloring algorithm \colrec, which initially assigns a random
color $\col(v)$ to each vertex $v \in V$. Each vertex $v$ maintains its
neighborhood $N(v)$ as a dynamic size array (\texttt{std::vector}). When
deleting an edge from this dynamic size array, we swap it with the last array
element and then remove it. As the edge set $E$ is empty initially, $\col$ is a
proper coloring. On insertion of an edge $e$, \colrec~adds the newly incident
vertices to the respective neighborhoods and checks whether their colors in
$\col$ are equal. If they are, we randomly choose a vertex $v$ incident to $e$,
and recolor $v$ with a random color $\col(v)$ from the color palette $\palette$.
We then check the neighborhood $N(v)$ and recursively recolor all neighboring
vertices whose color is equal to $\col(v)$. If the vertices have different
colors, the coloring remains proper and no action is taken. On an edge deletion,
the coloring also remains proper and thus we only update the affected
neighborhood data structures. This algorithm is, thus, very fast on edge
deletions and non-clashing insertions, but an edge insertion between two
vertices of the same color can cause expensive recursive cascading effects.

\subsection[Counting Fully-Dynamic Coloring]{Counting Fully-Dynamic $(\Delta + 1)$ Coloring.}

Another simple approach trys to avoid recursive cascading by keeping for each vertex $v \in V$ and color $\col$ a count of how many neighbors of $v$ have color $\col$.
In this \colcount~approach, we create an array
of size $\Delta + 1$ for each vertex $v$ and set all initial values set to $0$. 
This
array is used to count how many neighbors of $v$ have certain colors. We use the
same neighborhood data structure as previously. On edge deletion and
non-clashing edge insertion, \colcount~additionally updates the color counting
array. On insertion of edge $e$ with equal incident colors, we choose a random endpoint
$v$ of $e$. We then draw new random colors for $v$ until we
find a color that no neighbor of $v$ currently has. As the number of colors is
larger than the maximum degree, there is at least one such color. We then
recolor $v$ and update the color counting data structure.

\subsection{Constant-Time Algorithm of Henzinger and Peng.}

The algorithm of Henzinger and Peng~\cite{henzinger2020constant}
(\colhp) takes constant expected time per update against an \emph{oblivious adversary} and
is based on \emph{random vertex ranks}. Initially, each vertex $v \in V$ is
assigned a random rank $r(v) \in [0,1]$ and a random color $\xi(v) \in
\palette$. For each vertex $v \in V$, the set of neighbors is partitioned into
$H_v$, the neighbors with rank $> r(v)$ and $L_v$, the neighbors with rank $<
r(v)$. We implement $L_v$ and $H_v$ as hash sets. The algorithm aims to maintain
a proper $(\Delta + 1)$-coloring. Edge deletions or non-clashing edge insertions
do not lead to a violation of a proper coloring, only the insertion of edge $e =
(u,v)$ with $\xi(u) = \xi(v)$ triggers a recolor of vertex $v$, where $r(v) >
r(u)$. \colhp~samples in a carefully designed procedure a random new color 
 that is either not used by any
neighbor of $v$ or only by a single neighbor $w$ with $w \in L_v$, and assigns this
color to $v$. If such a ``conflicting'' neighbor $w$
exists, it triggers a recoloring of $w$. 

\subsection{Constant-Time Algorithm of Bhattacharya et al.}

The algorithm of Bhattacharya et al.~\cite{bhattacharya2019fully}
(\colbgkls) 
takes constant expected time per update against an \emph{oblivious adversary} and
partitions the vertex set into  a \emph{hierarchy} of $\log{\Delta}$ levels, where vertices
rise in level if they have many neighbors of lower rank. Similar to the other
algorithms, edge deletions or non-clashing edge insertions do not lead to a
violation of a proper coloring, only clashing edge insertions trigger a recolor
of the vertex that was more recently recolored. When a vertex $v \in V$ is
recolored, \colbgkls~picks a new color for $v$ with level $l(v)$ in the
following way: if $v$ has fewer than $3^{l(v) + 2}$ neighbors of lower or equal
level, we set $l(v) = -1$ and pick the new color $\xi(v)$ from the set of colors
that no neighbor of $v$ has. Otherwise, the level of $v$ is raised to the lowest
level $l(v)$, so that $v$ has fewer than $3^{l(v) + 2}$ neighbors of lower or
equal level and then pick a random new color that is either not used by any
neighbor of $v$ or by exactly one neighbor $w$, where the level of $w$ is lower
than the level of $v$. If there is such a neighbor $w$, the recoloring algorithm
is then called on $w$.

\section{Algorithms for the Fully Dynamic Maximal Matching Problem}
\label{ss:alg-match}

We now briefly outline the implemented algorithms for the fully dynamic maximal
matching problem. We implement a trivial maximal matching algorithm, as well as
the level set partitioning algorithms of Solomon~\cite{solomon2016fully} and
Baswana et al.~\cite{baswana2015fully}; and the random rank algorithm of
Behnezhad et al.~\cite{behnezhad2019fully}. For the algorithm of Baswana et
al., we use an implementation of Henzinger et al.~\cite{henzinger2020dynamic}.
Additionally, we introduce a variation of the algorithm of Behnezhad et al. that
does not guarantee a polylogarithmic running time but performs much faster in
practice.

\subsection{Trivial Fully Dynamic Maximal Matching.}

A matching is maximal if for each edge $e=(u,v) \in E$ at least one of the
vertices $u$ and $v$ is matched. Initially $E$ is empty and therefore the empty
matching is maximal. In the trivial fully dynamic maximal matching algorithm
\matchtriv, each vertex $v \in V$ maintains its neighborhood $N(v)$ as a dynamic size array (\texttt{std::vector}). For the insertion of edge $e = (u,v)$, we check whether $u$ and $v$ are
both unmatched and if they are, we add $e$ to the matching. Otherwise we leave
the matching unaltered, as the new edge already has an incident edge in the
matching. On deletion of a matching edge $e = (u,v)$, we have to remove $e$ from
the matching (as it no longer exists) and, thus, edges incident to $u$ or $v$
can be uncovered. In this case \matchtriv\ iterates over the neighborhoods of the
newly unmatched vertices $u$ and $v$ and checks whether new matching partners
can be found.

\subsection{Algorithms of Solomon and Baswana et al.}

The algorithm of Baswana et al.~\cite{baswana2015fully} (\matchbgs) and later
the algorithm of Solomon~\cite{solomon2016fully} (\matchsolomon) takes $O(\log
n)$, resp.~ $O(1)$ expected time per update against an \emph{oblivious
adversary} and both partition the vertex set into a \emph{hierarchy} of
$\log{n}$ levels, where intuitively a vertex rises in levels if it has
many neighbors in lower levels. A vertex $v$ is in level $l(v) = -1$, if it is
unmatched and in a level $l(v) \geq 0$ if it is matched. Each vertex maintains a
set of neighbors $O_v$ of neighbors of lower level as well as a set of
higher-level neighbors $I_v[l]$ for each level $l(v) \leq l < \log{n}$. For
insertion of edge $e = (u,v)$, the algorithms first add the edge to the
respective data structures and check whether both $u$ and $v$ are unmatched
(i.e. are in level $-1$) and if they both are, they add $e$ to the matching and
elevate both incident vertices to level $0$. If at least one of $u$ and $v$ are
already matched, no further action is taken. For deletion of edge $e = (u,v)$,
the algorithms remove the edge from the respective data structures. If $e$ is
not a matching edge, the vertex levels remain constant, otherwise $u$ and $v$
are now unmatched vertices and the algorithms aim to find new matching partners:
a vertex $v$ with few vertices in $O_v$ checks whether there exists an
unmatched $w \in O_v$, and if it exists, adds $(v,w)$ to the matching and sets
the levels of $v$ and $w$ to $0$. If no such $w$ exists, $v$ remains unmatched.
For a vertex $v$ with many neighbors in $O_v$, they instead pick a random
(potentially matched) neighbor $w \in O_v$ and adds $(v,w)$ to the matching. If
a matching edge $(w,w')$ exists, it is removed from the matching and the
algorithms use the same algorithm to try and find a new matching partner for the
now unmatched vertex $w'$. The algorithm of Baswana et
al.~\cite{baswana2015fully} explicitly maintains the vertex levels and achieves
an amortized runtime of $\Oh{\log n}$, the algorithm of
Solomon~\cite{solomon2016fully} achieved constant running time by lazily
maintaining the vertex levels and therefore expensive waves of vertex level
updates. We implement the neighborhood sets $O_v$ and $I_v[l]$ as hash sets.

\subsection{Algorithm of Behnezhad et al.}

Behnezhad et al.~\cite{behnezhad2019fully} give a random rank based algorithm
(\matchbbhss) for the fully dynamic maximal independent set problem and extend
that algorithm to the fully dynamic maximal matching problem. The resulting
algorithm takes $\Oh{\log^2{\Delta}\log^2{n}}$ expected time per update against
an \emph{oblivious adversary} and assignes during preprocessing a \emph{random
rank} to each edge, which induces a random order on the edges. It then maintains
the \emph{lexicographically first maximal matching} over this order of the
edges. 

A matching $\match$ is the lexicographically first maximal matching (LFMM) of a graph
$G=(V,E)$ with a ranking $\pi: E \rightarrow [0,1]$ over the edges if it is
equal to the matching created by the following greedy
algorithm~\cite{behnezhad2019exponentially, behnezhad2019fully} that finds a
maximal matching. Initially, all edges are white. We iteratively pick the white
edge $e$ with minimum rank $\pi(e)$, add it to the matching, and then make $e$ and all
its incident white edges black. We call this matching  $LFMM(G,\pi)$. In this
greedy algorithm, every edge is either added to the matching or turned black by an
incident edge of smaller rank $\pi$. In the following we say an edge $e_1 \in M$
\emph{covers} an edge $e_2 \not \in M$, if they are incident and $\pi(e_1) <
\pi(e_2)$. An edge $e \not \in M$ that does not have any matching edge which
covers $e$ is called \emph{uncovered}.

The algorithm of Behnezhad et al.~\cite{behnezhad2019fully} draws for each
inserted edge $e$ a random rank $\pi(e)$ in $[0,1]$. For each edge $e$ it also
maintains its eliminator rank $k(e) \leq \pi(e)$, where $k(e)$ is the rank
of the edge incident to $e$ that is in the matching. For a matching edge $e \in
M$, $r(e) = k(e)$, otherwise $k(e) < \pi(e)$. Each vertex $v$ maintains a
self-balancing binary search tree of its incident edges ordered
by their eliminator rank. This is necessary as it needs to be able to find all incident edges below a given eliminator rank. Each vertex also maintains its vertex rank $k(v)$ which is
equal to the  edge rank  of the incident matching edge if one exists and $\infty$ otherwise.
For the insertion of edge $e = (u,v)$, we add $e$ to the respective data structures
and check whether the edge should be part of the LFMM (\ie $\pi(e) < k(u)$ and
$\pi(e) < k(v)$). If it is, we add $e$ to the matching and remove the matching
edges incident to $u$ and $v$. The algorithm then updates the eliminator ranks
of all affected edges in the search tree data structures and finds new partners
for the newly unmatched vertices. For the deletion of edge $e = (u,v)$, the
algorithm removes $e$ from the respective data structures and if $e$ was in the
LFMM, checks all incident edges $e'$ with $k(e') \geq \pi(e)$ to find new
matching partners for $u$ and $v$. Similar to the edge insertion, the algorithm
updates the search tree data structure and finds new matching partners for all
vertices that become unmatched over the course of the algorithm. The algorithm has amortized expected $O(\log^4 n)$ update time.

\subsection{New Random-Rank Algorithm for Maximal Matching.}

The algorithm of Behnezhad et al.~\cite{behnezhad2019fully} performs many updates in 
search tree data structures. As these updates are very expensive in practice, we give an alternative random-rank algorithm
(\matchrr) based on their idea that does not need to perform these expensive
updates. This new algorithm also maintains the lexicographically first maximal
matching on a dynamic graph and is much faster in practice despite having a
worse asymptotic running time bound. In the following we give a detailed
description of this algorithm.

The data structures used in \matchrr~are closely related to the algorithm of
Behnezhad et al.~\cite{behnezhad2019fully}. For each inserted edge $e$,
\matchrr~draws a random edge rank $\pi(e) \in [0,1]$. Each vertex maintains
$k(v)$, the rank of the incident matching edge if it exists - if $v$ is unmatched,
$k(v) = \infty$, the set of neighbors $N(v)$ and the matched neighbor $P(v)$,
for an unmatched vertex $v$, $P(v) = -1$. In contrast to the algorithm of
Behnezhad et al., the neighbors are ordered by their ranks and \emph{not} by their
eliminator rank and, thus, do not need to be updated when the matching changes.
We next give a detailed description of the algorithm.

\paragraph{Algorithm description.}

Similar to the algorithm of Behnezhad et al.\cite{behnezhad2019fully}, we aim to
maintain the lexicographically first maximal matching in $G$ according to a
random edge ranking $\pi$. The set of vertices is static, the edge set is
initially empty. On edge insertion, we call Algorithm~\ref{alg:insert}
(pseudocode in Appendix~\ref{app:pseudocode}) and on edge deletion, we call
Algorithm~\ref{alg:delete} (pseudocode in Appendix~\ref{app:pseudocode}), which
each use Algorithm~\ref{alg:partner} (findNewPartners()) to find new partners
for vertices whose previous partner was matched with a new matching partner in
order to maintain the invariant that all non-matching edges are covered.

\begin{algorithm}[t!]
    \caption{findNewPartners(): Find new partners for unmatched vertices, cover all uncovered edges \label{alg:partner}}
    \begin{algorithmic}[1]
        \INPUT $S \leftarrow$ vertex priority queue

        \While{$S$ is not empty}
        \State $v,r_v \leftarrow$ remove lowest rank vertex in S
        \For{$w \in N(v)$ with $r_v < \pi(v,w) < k(v)$}
            \If{$\pi(v,w) < k(w)$} \Comment{Add $(v,w)$ to $\match$}
                \If{$k(w) < \infty$}
                    \State $x \leftarrow P(w)$
                    \State S.insert($x, k(w)$)
                    \State $P(x) \leftarrow -1$
                    \State $k(x) \leftarrow \infty$
                \EndIf
                \If{$k(v) < \infty$}
                    \State $x \leftarrow P(v)$
                    \State S.insert($x, k(v)$)
                    \State $P(x) \leftarrow -1$
                    \State $k(x) \leftarrow \infty$
                \EndIf
                \State $P(v) \leftarrow w$
                \State $P(w) \leftarrow v$
                \State $k(v) \leftarrow r_e$
                \State $k(w) \leftarrow r_e$
                \State \textbf{return}
            \EndIf
        \EndFor 
        \EndWhile       
    \end{algorithmic}
\end{algorithm}

On insertion of edge $e = (u,v)$, Algorithm~\ref{alg:insert} first draws a
random rank $\pi(e)$ and then inserts the edge in the appropriate data
structures. We make sure $\pi(e)$ is unique by re-drawing the random rank if
another edge already has equal rank. If $\pi(e)$ is smaller than the minimum of
$k(v)$ and $k(u)$, $e$ is added into the matching. If $u$ or $v$ had previous
matching partners, they are removed from the matching and added into a priority
queue $S$ with priority equal to the rank of their previous matching edge $k(u)$
resp. $k(v)$. We then call Algorithm~\ref{alg:partner} to make sure all edges
are covered again, i.e., all vertices in $S$ are processed. 

Algorithm~\ref{alg:partner} takes vertex $v$  of $S$ with lowest
priority, called $r_v$, and checks whether there exists an incident edge $e = (v,w)$ with
priority $r_v < \pi(e) < k(v)$ and $\pi(e) < k(w)$ that is added to LFMM. 
Note that $r_v$ is the rank of the last edge with which $v$ was matched at the time that it was added to $S$.
If $v$ 
(and $w$ analogously) is currently matched, i.e. $P(v) \neq -1$, we add $P(v)$
to $S$ with priority $\pi((v,P(v)))$, as all edges incident to $v$ with priority $>
\pi((v,P(v)))$ are now potentially uncovered. Note that even though vertex $v$
was unmatched when it was added to priority queue $S$, it is possible that an
 edge incident to $v$ was added to the matching in the meantime and, thus, $v$ is
currently matched. We repeat this process until priority queue $S$ is empty.

On deletion of an edge $e = (u,v)$, we delete it from the neighborhoods of $u$
and $v$ and check whether it is a matching edge. If it is, we set both $u$ and
$v$ to unmatched and add them to $S$ with a starting rank of $\pi(e)$ to check
every edge with a rank of $> \pi(e)$. Similar to the insertion algorithm, we
then call Algorithm~\ref{alg:partner}.

In the following we show that this algorithm terminates and correctly maintains
a LFMM on edge insertion and deletion. We first prove Lemma~\ref{lem:killed},
which says that the LFMM is the only matching, in which each non-matching edge
is covered by an incident matching edge of smaller rank.

\begin{lemma} \label{lem:killed} A matching $\match$ is $LFMM(G,\pi)$ if and
only if for each edge $e = (u,v) \in E$ either $e \in \match$ or $\exists w \in
N(u) \cup N(v)$ s.t. $(u,w) \in \match$ with $\pi(u,w) < \pi(e)$ or $(v,w) \in
M$ with $\pi(v,w) < \pi(e)$. (Proof in Appendix~\ref{app:proofs})
\end{lemma}

\begin{claim} 
    \label{claim:onlylarger}
    In Algorithm~\ref{alg:partner}, when vertex $v$ is removed from $S$ with
    rank $r_v$, then any new vertex $w$ that is added to $S$ during this
    iteration of Algorithm~\ref{alg:partner} is added with rank $> r_v$. (Proof in Appendix~\ref{app:proofs})
\end{claim}

\begin{claim}
    \label{claim:non-decreasing}
    Algorithm~\ref{alg:partner} checks vertices in non-decreasing priority
    order. (Proof in Appendix~\ref{app:proofs})
\end{claim}

\begin{claim} 
    \label{claim:terminate}
    Algorithm~\ref{alg:partner} terminates after at most $m$ iterations of the while-loop. (Proof in Appendix~\ref{app:proofs})
\end{claim}
Each iteration takes time $O(n)$ leading to a $O(nm)$ worst-case time per operation.
To show the correctness of the algorithm we will show that before each call to Algorithm~\ref{alg:partner}, the following two invariants hold:
\begin{description}
\item[I1] At least one endpoint $v$ of each uncovered edge $e$ is in $S$ (or currently
processed in Algorithm~\ref{alg:partner}) and $\pi(e) > r_v$,
\item [I2] All edges of rank at most the lowest rank in $S$ are covered.
\end{description}
First, however, we show the following lemma:

\begin{lemma}\label{lem:partner}
If I1 and I2 hold at the beginning of Algorithm~\ref{alg:partner} then at its termination  the matching is a lexicographically first maximal matching.
\end{lemma}
\begin{proof}
Assume I1 and I2 hold when Algorithm~\ref{alg:partner} is called.
The algorithm removes repeatedly the lowest priority vertex $v$, say with
priority $r_v$, from $S$ and check in increasing rank order (starting from $r_v$) whether the edges incident to $v$ are covered. All uncovered edges incident to $v$ must have rank
$> r_v$ by our invariant.
If an edge $e = (v,w)$  with $\pi(e) > r_v$ is not covered
when it is checked, we add $e$ to
the matching, which immediately covers all remaining unchecked edges incident $v$.
Thus, all edges incident to $v$ are now covered.
Then we evict the current matching edges incident to $v$ and $w$ from
$\match$. Note that both must have rank $>\pi(e)$, but both are now covered. As their partners might now have incident uncovered edges of rank $>
\pi(x,P(x)) > \pi(e)$ and $> \pi(y,P(y)) > \pi(e)$, we add $P(x)$ and $P(y)$ to the priority queue
$S$ with priority $\pi(x, P(x))$, resp.~$\pi(y,P(y))$. This maintains Invariant I1.
It also maintain Invariant I2 as (a) all new potentially uncovered edges are
incident to either $P(x)$ or $P(y)$ and have rank $>\min\{\pi(x, P(x)),\pi(y,P(y))\}$, which is larger than the minimum rank in $S$, (b) all edges incident to $v$ are now
covered, and (c) by the Invariant 1 for all other uncovered edges $e'$ there exists an endpoint $u$ in $S$ with $\pi(e') > r_u$. 
As the queue $S$ is guaranteed to be empty after the execution of
Algorithm~\ref{alg:partner} (Claim~\ref{claim:terminate}), it follows that all
edges are covered, which implies by Lemma~\ref{lem:killed} that the matching is a 
lexicographic first maximal matching.
\end{proof}

\begin{lemma} \label{lem:insert-correct} Given the matching $\match =
LFMM(G,\pi)$ and the insertion of an edge $e = (u,v)$,
Algorithm~\ref{alg:insert} maintains the lexicographically first maximum
matching $LFMM(G+e,\pi)$. (Proof in Appendix~\ref{app:proofs})
\end{lemma}

\begin{lemma} \label{lem:delete-correct} Given the matching $\match =
     LFMM(G,\pi)$ and the deletion of an edge $e = (u,v)$,
     Algorithm~\ref{alg:delete} maintains the lexicographically first maximum
     matching $LFMM(G-e,\pi)$. (Proof in Appendix~\ref{app:proofs})
\end{lemma}

\section{Experiments and Results} \label{s:experiments}

\subsection{Experimental Setup and Methodology.}

We implemented the algorithms using \CC-20 and compiled all codes using
g++-11.1.0 with full optimization (\texttt{-O3}). Our experiments were conducted
on a machine with two Intel Xeon E5-2643 v4 with 3.4GHz with 6 CPU cores each
and 1.5 TB RAM in total. We performed five repetitions per instance and report
average running~time. In this section we first describe the experimental
methodology. Afterwards, we evaluate different algorithmic choices in our
algorithm and then we compare our algorithm to the state of the art. When we
report a mean result we give the geometric mean as problems differ strongly in
result and time. All codes are sequential.

\subsubsection{Instances.}
\label{exp:instances}

For our experiments we used both both real-world and generated instances.
For generated instances, we use the \texttt{KaGen} graph
generator~\cite{funke2019communication}, which generates a wide variety of
static graph families. We use a family of random Erd\H{o}s-R\'enyi graphs and a
family of random hyperbolic geometric graphs to create our fully dynamic
instances.

An \emph{Erd\H{o}s-R\'enyi graph (ER graph)} $G(n,p)$ is a random graph with $n$
vertices where each two vertices are connected by an edge with probability $p$,
independently from the other edges in the graph. The graph has on average
${n\choose 2}\cdot p$ edges and a Poisson degree distribution.

\emph{Random Hyperbolic Geometric Graphs (RHG
graphs)}~\cite{krioukov2010hyperbolic} are random graphs that replicate many
features of social graphs~\cite{chakrabarti2006graph}: the degree distribution
follows a power law, they have small diameter and often exhibit a community
structure. These graphs are geometric graphs on a disk in hyperbolic space, in
which nodes that are close to each other (in hyperbolic space) are connected by
an edge. In these graphs, the nodes that are close to the center of the disk
have a very high degree, less central vertices usually have a much lower degree.

We generated static ER and RHG graphs with varying number of vertices and edges
and used different techniques to create fully dynamic instances from the static
graphs. We explain the dynamization methods when we report about the respective
experiments. On the ER graphs, the maximum degree $\Delta$ is generally
$20-40\%$ larger than the average degree, on the RHG graphs, $\Delta$ is
generally a constant factor of $n$, \ie there are vertices (close to the center
of the hyperbolic plane) with a very high degree.

Additionally, we used as data sets six real-world Wikipedia update
sequences that were obtained from the Koblenz Network Collection
KONECT~\cite{kunegis2013konect}. The graphs are directed graphs and given as a
sequence of individual edge insertions and deletions and model the link
structure between Wikipedia articles. While the instances are given as directed
graphs, we interpret them as undirected graphs, \ie while there exists either an
edge from $u$ to $v$ or from $v$ to $u$, our graph instance has an undirected
edge $(u,v)$. The instances are initially empty and have between $100K$ and
$2.2M$ vertices and between $1.6M$ and $86M$ edge updates. As these graphs have
very high maximum degree $\Delta$ and, thus, our implementations of the dynamic
coloring algorithms \colhp,~\colcount~and~\colbgkls~require space that is
quadratic in $n$, we only run experiments on them for the fully dynamic matching
problem.

\subsection[Fully Dynamic Coloring]{Fully Dynamic $(\Delta + 1)$ Coloring.}
\label{exp:color}

We implemented all algorithms described above, namely the random rank algorithm of Henzinger and
Peng~\cite{henzinger2020constant} (\colhp), the hierarchical algorithm of
Bhattacharya~\etal\cite{bhattacharya2019fully} (\colbgkls), as well as two
trivial algorithms (\colcount~and~\colrec) for the problem. 

The non-trivial algorithms
\colhp~and \colbgkls~use Google Dense Hash Set~\cite{web:googledense} to implement the neighborhood data structures efficiently, while the trivial
algorithms \colcount~and~\colrec~use dynamic size arrays (\texttt{std::vector}). For the trivial algorithms we can use dynamic size arrays, as updates generally affect all neighbors of a vertex and so we need to iterate over the whole neighborhood of a vertex. In \colhp~and~\colbgkls~there are more neighborhood queries (\ie is vertex $u$ a neighbor of vertex $v$?), which can be answered in expected constant time using a hash table.

\subsubsection{Random Update Sequence on Generated Graph.}
\label{exp:c-random}

For this experiment, we generated ER and RHG graphs with varying number of
vertices $n$ and number of edges $m$. For both we generated two
sets of graphs, in set (1) the average density remains constant while the number
of vertices changes; while in set (2) the number of vertices is constant while
the average density changes. In set (1), we generate graphs with $n = 2^{13},
2^{14}, 2^{15}, 2^{16}, 2^{17}, 2^{18}$ and $m = n \cdot 256$, whereas in set (2),
the number of vertices is constant $n=2^{16}$, while $m = n \cdot
\{64,128,256,512,1024\}$. For each of these $22$ graph instances, we generate
random update sequences that start with an empty graph with $n$ nodes and no
edges and insert the whole edge set $\mathcal{S}$ in random order, interspersed
with deletions of previously inserted edges. Let $\rho \in [0,1]$ be the
deletion rate. The random update sequence is generated as follows: first we
randomly shuffle the edges in $\mathcal{S}$, afterwards with probability
$\frac{1}{1 + \rho}$ we insert the next edge in $\mathcal{S}$ into $G$ or with
probability $\frac{\rho}{1 + \rho}$ we remove a random edge (if one exists) from
$G$. This process is repeated until every edge in $\mathcal{S}$ is inserted,
resulting in a total of $|\mathcal{S}|$ edge insertions and about $\rho
|\mathcal{S}|$ edge deletions. For each graph, we create an update sequence each
with $\rho = \{0, 10\%, 25\%, 50\%, 75\%\}$, which results in a total of $220$
graph instances. Note that $\rho = 0$ indicates that there are no deletions. Insertions are random and not dependent on earlier
insertions. Recall that the algorithms have to recolor a vertex whenever we insert an edge $e =
(u,v)$ where $\col(u) = \col(v)$, due to the random nature of the insertions and the fact
that there are $\Delta + 1$ available colors
this happens with a probability of roughly $\frac{1}{\Delta}$ independent of
the choice of the algorithm and the experiment shows average update times on
random updates. As the RHG graphs have very high maximum degree $\Delta$ and
therefore data structures grow very large, we exclude the initialization time in
this experiment; so we can report the update performance of the algorithms and
not how much time is required to initialize large arrays and hash tables. 

\begin{figure}[t!]
    \includegraphics[width=.49\textwidth]{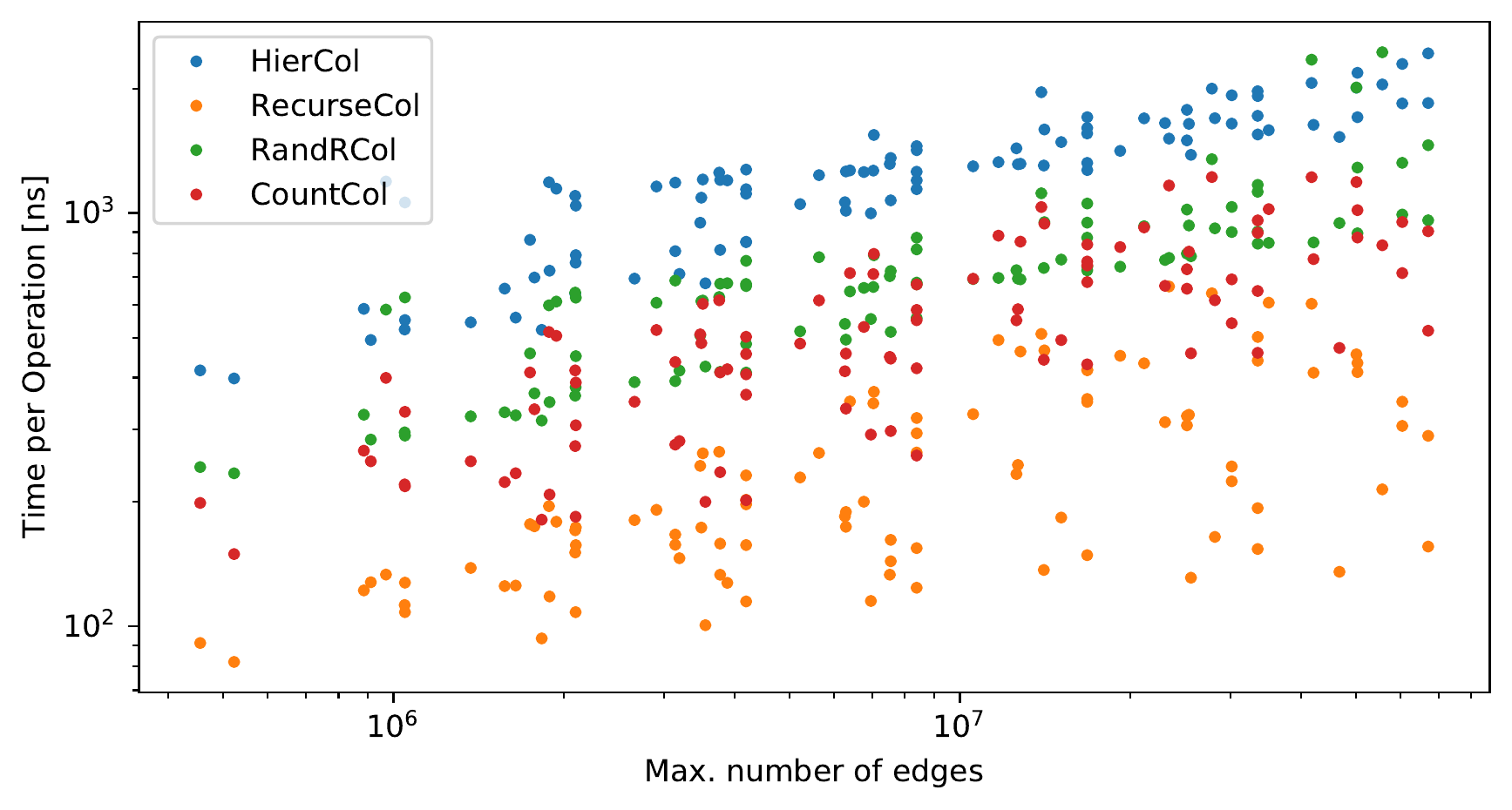}
    \caption{\label{fig:c-scaling} All random insertion and deletion sequences}
\end{figure}

Figure~\ref{fig:c-scaling} shows the results for all random insertion and
deletion sequences. We can see that in these random instances, the simple
\colrec~algorithm is fastest for almost all instances, followed by \colcount~and \colhp, which is on average about $30-50\%$ faster than \colbgkls. \emph{Interestingly, there is
no major effect of density (i.e., set (1) vs.~set (2)), graph type and deletion rate $\rho$ on the relative performance of
the algorithms, generally \colhp~performs slightly better than \colcount~on
denser instance with a higher deletion rate $\rho$.} 
For all
algorithms we can see that the average time per operation  increases slightly on
larger instances, as the basic operations are slightly slower on larger arrays
and hash tables.

For all algorithms, the majority
of the running time is not spent in recoloring, but in  the insertion and removal of
vertices from neighborhood data structures and the use of different neighborhood data structures has a major impact on the performance: While finding a
neighboring vertex can take up to $\Oh{\Delta}$ using arrays, the constant time
operations such as insertion and deletions are significantly faster on a hash table array compared to dynamic arrays. Specifically, on a data structure benchmark on the same hardware  inserting 1M (random integer) elements (into an empty resp.~already filled with 10M elements data structures) into a hash table vs.~a dynamic size array is a factor 10, resp.~13 faster,
accessing all elements sequentially is roughly a factor of 20 faster, and deleting 1M elements is roughly a factor of 100 faster. 

One explanation for the superiority of \colhp~over~\colbgkls~is that
\colhp~only updates neighborhood counts on the lower-rank side of a deleted edge,
while \colbgkls~counts the color of the incident vertex on both sides.
Therefore, edge deletions are faster in \colhp. Additionally, \colbgkls~changes
vertex levels in the hierarchy over the course of the algorithm, which induces
additional data structure insertions and deletions. In this experiment,
according to performance profiles \colbgkls~spends $~20-25\%$ of the time spent
in edge insertions in level changes, \ie moving neighboring vertices to
different data structures for the respective hierarchy levels. The vertex
ranking of \colhp~remains constant over the course of the algorithm, so
neighborhood data structures only need to be updated when edges are inserted or
deleted.

Performance profiles also indicate that the more complex updating schemes of
\colbgkls~and~\colhp~have a much larger impact on the running time compared to
the simple iteration over neighboring vertices in the trivial algorithms.
\colhp~spends about $25-30\%$ of the running time of insertions in recoloring
operations, \colbgkls~about $35-40\%$ of insertion time, whereas the trivial
algorithms rarely exceed $10\%$. For all algorithms, the vast majority of the
remaining running time is spent in adding and removing neighboring vertices from
the basic neighborhood data structures. Note that these numbers are derived from
sampling-based performance profiles and also vary depending on the instance.
They should therefore be seen as estimates, however they clearly indicate that
the simple recoloring scheme of trivial algorithm is significantly faster than
more elaborate recoloring approaches.

\subsubsection{Incremental Clashing Sequence.}
\label{exp:c-clash}

In this experiment we generate incremental (insertions-only) worst-case update sequences for both of the trivial algorithms
\colcount~and \colrec, by selecting two random vertices $u$ and $v$ with 
$\col(u) = \col(v) = \col$ and then add an edge $e = (u,v)$ to the graph, forcing the algorithm to recolor one of the two. Note that this creates two different worst-case updates
sequences as the two algorithms generate different colorings. 
Note that both \colhp~and \colbgkls~achieve constant amortized worst-case running time against an \emph{oblivious adversary} (that must generate all updates without knowing the colors assigned to the vertices)  by making sure that the palette of new possible
colors when recoloring has $\Theta(\Delta)$ colors and choosing a color from it randomly. However, their running time guarantees do \emph{not} hold against an \emph{adaptive} adversary as the one above.

\begin{figure}[t!]
    \includegraphics[width=.49\textwidth]{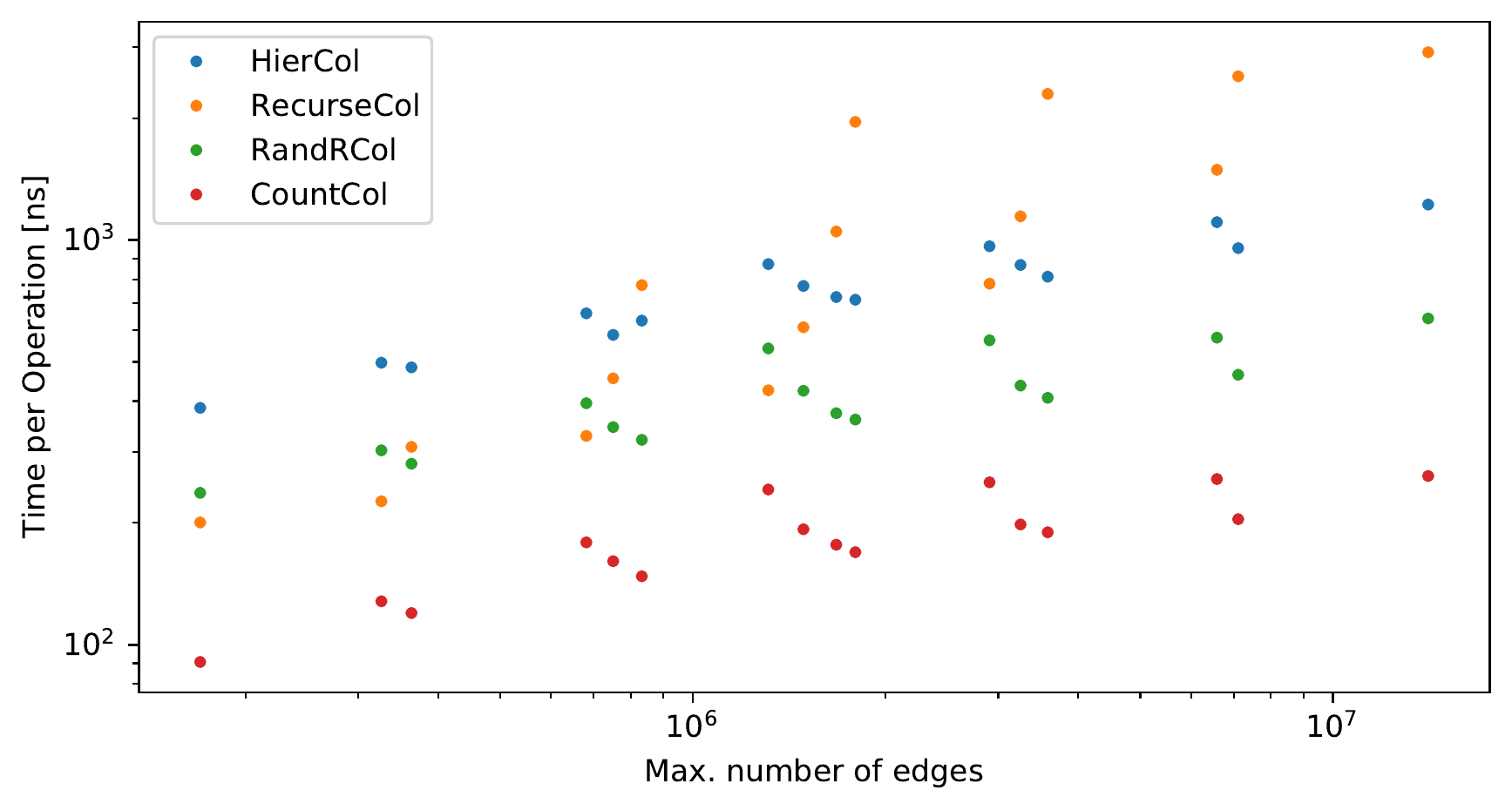}
    \caption{\label{fig:cc-recursive} Clashing Sequence for \colrec}
\end{figure}

For either algorithm we generated instances with $n = 2^{12}, 2^{13}, 2^{14},
2^{15}$ and a maximum degree of $128, 256, 512, 1024$, \ie a total of $16$
instances per algorithm. 
Figure~\ref{fig:cc-recursive} shows the clashing sequence for \colrec~and
Figure~\ref{fig:cc-counting} shows the clashing sequence for \colcount. Note
that in these figures, both the x- and the y-axis are logarithmic. We can see the expected
result: the time per operation in the clashing algorithm increases quickly achieving the worst performance over all algorithms. Furthermore, the running time of \colcount~increases on its worst-case sequence much faster (reaching over $10^4$ ns per operation after $10^6$ operations)  than \colrec~on its respective sequence. The reason is that
when recoloring,
\colcount~updates the neighborhood color counts of \emph{all} neighbors of the
recolored vertex, while \colrec~recolors only one or few vertices and checks
that no neighbors have the same color. While this also takes time linear to the
vertex degree, it is much faster in practice, as there are no write operations
for all neighbors of the recolored vertex. Note also that for reasonably large
$n$ both \colhp~and~\colbgkls~perform better than the clashing algorithm and
\colhp~always performs better than \colbgkls.

\subsubsection{Vertices of Almost Equal Degree.}
\label{exp:c-samedeg}

In Sections~\ref{exp:c-random}~and~\ref{exp:c-clash} we can see that the
simplicity of \colrec~makes the algorithm very fast when there are enough free
colors for many vertices and the recursive cascading does not recurse too
deeply. However, the situation changes when the average vertex degree is very
close to $\Delta$ and an average vertex does not have many free colors in its
neighborhood. In this experiment, we first generate random initial instances
$G_i$ with a maximum degree of $(\Delta - 1)$ by adding random edges between
vertices of degree $< (\Delta - 1)$ until the average degree is very close to
$(\Delta - 1)$. We then perform a random update sequence on top of this graph,
where we overlay this graph with a graph $G_o$ of maximum degree $1$, which
results in a graph $(G_i + G_o)$ of maximum degree $\Delta$. For each update, we
draw a random vertex $v$ and check whether it has a neighbor $u$ in $G_o$. If it
does, we remove $(v,u)$ from the graph, otherwise we draw random vertices until
we find a vertex $u$ that does not have a neighbor in $G_o$ and $(v,u) \not \in
E(G_i)$, and add $(v,u)$ to the graph. Therefore, the average degree is very
close to the maximum degree over the whole edge update sequence. We generate
$G_i$ with $n = 2^{13}, 2^{14}, 2^{15}, 2^{16}$ and a maximum degree of $512,
1024, 2048, 4096$ and perform $10^7$ random updates. 
\begin{figure}[t!]
    \includegraphics[width=.49\textwidth]{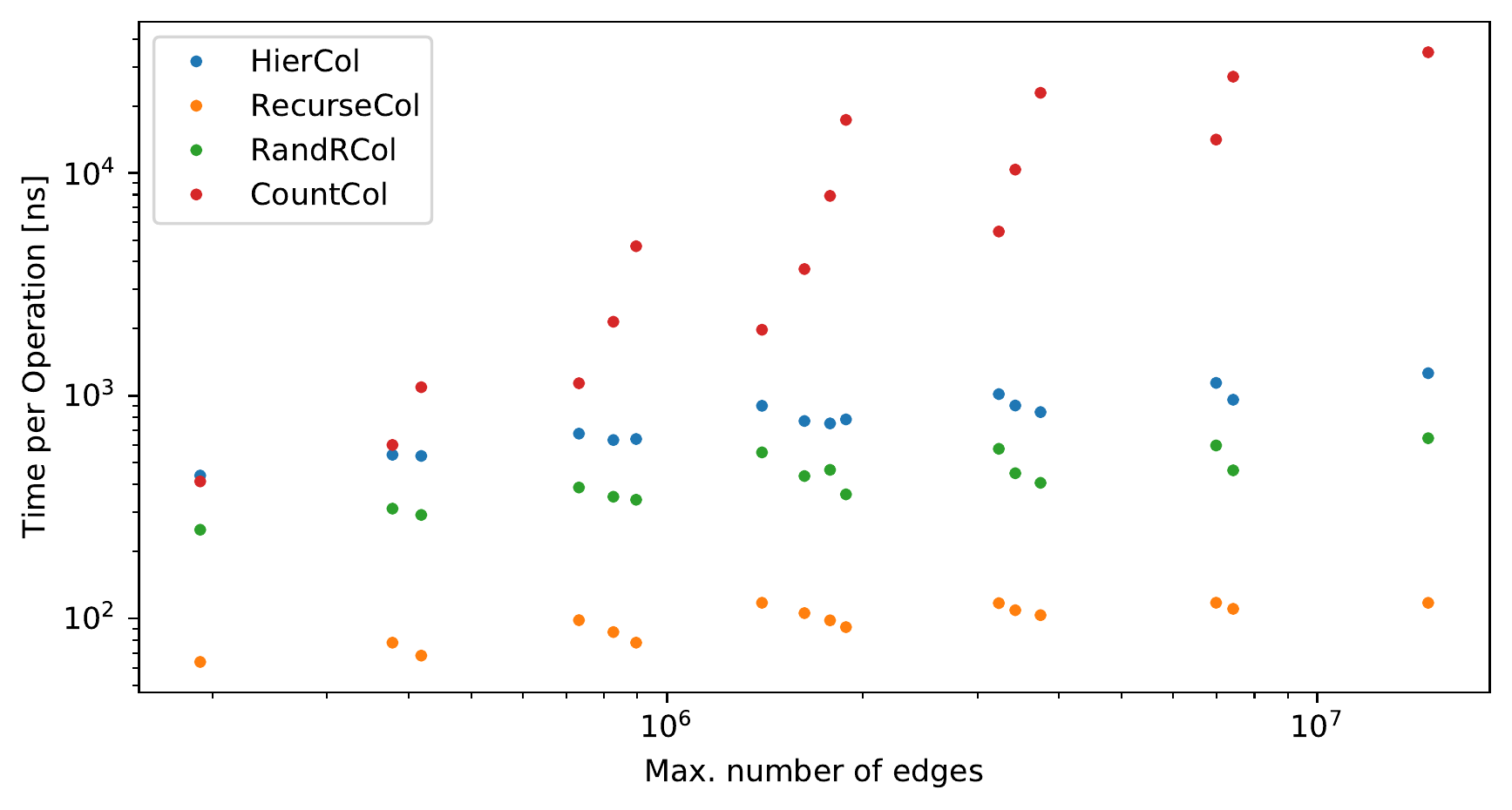}
    \caption{\label{fig:cc-counting} Clashing Sequence for \colcount}
\end{figure}

In Figure~\ref{fig:cc-samedeg} we can see that \colrec~performs worse compared to other algorithms when the
average degree is very close to the maximum degree $\Delta$ and the recursion
depth grows. On the instances, in which the average degree is $>0.999\Delta$,
the algorithm is significantly slower than \colcount, which does not experience slowdown when the average degree is close to the maximum degree. The two non-trivial algorithms also do not experience slowdown when the average degree is almost as large as the maximum degree.

We also generated graphs $G_i$ in which we made sure that every vertex has equal degree $(\Delta - 1)$. When performing this experiment on those graphs, \colrec~often runs into non-ending recursive calls where the recursion depth keeps growing until the recursion stack overflows. Therefore, in Figure~\ref{fig:cc-samedeg} we use graphs where the average degree is \emph{almost equal} to the maximum degree.

\begin{figure}[t!]
    \includegraphics[width=.49\textwidth]{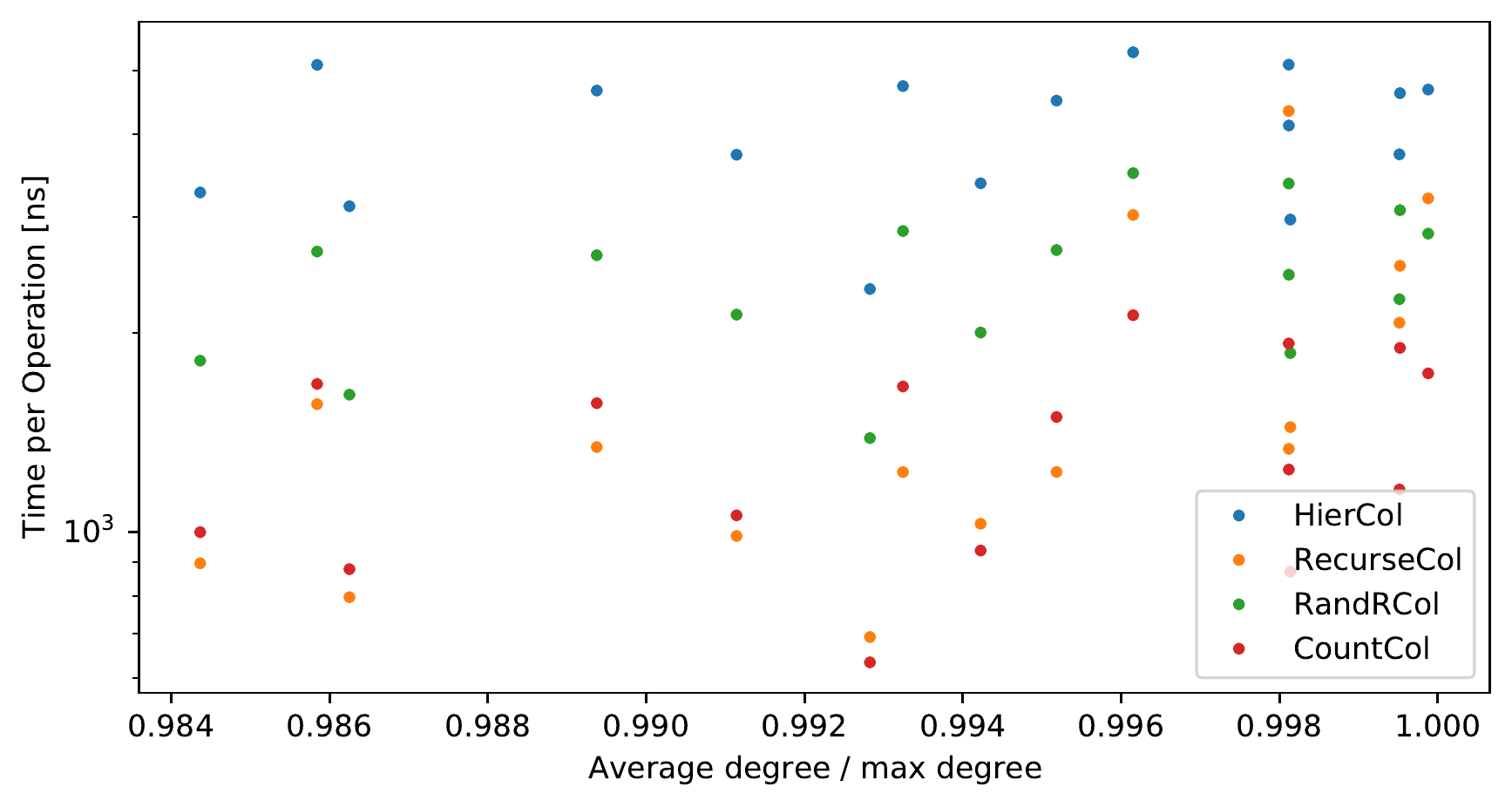}
    \caption{\label{fig:cc-samedeg} Vertices of almost equal degree}
\end{figure}

\subsection{Fully Dynamic Maximal Matching.}

For the fully dynamic maximal matching problem, we implemented the algorithms of
Behnezhad et al.~\cite{behnezhad2019fully} (\matchbbhss), the algorithm of
Solomon~\cite{solomon2016fully} (\matchsolomon), our random-rank algorithm described in detail in
Section~\ref{ss:alg-match} (\matchrr) and the trivial algorithm  described in
Section~\ref{ss:alg-match} (\matchtriv), 
and we used an implementation of Henzinger~\etal\cite{henzinger2020dynamic} of
the algorithm of Baswana et al.~\cite{baswana2015fully} (\matchbgs).

The hierarchical algorithms \matchbgs~and \matchsolomon~use Google Dense Hash
Set~\cite{web:googledense} for neighborhood data structures so that they can
perform neighborhood queries in expected constant time, the trivial algorithm
\matchtriv~uses dynamic size arrays (\texttt{std::vector}). As \matchbbhss~needs
to keep all neighbors in order and often updates this order, it uses a
self-balancing binary tree data structure (\texttt{std::set}). Compared to
Google Dense Hash Set, this self-balancing binary tree is slower by a factor of
12, resp.~25 on insert of 1M random elements in an empty (or pre-filled with 10M
elements) data structure, by a factor of 11 on sequential access and by a factor
of 30, resp.~35 on deletion of 1M elements. \matchrr~has neither neighborhood queries nor needs to update the order of
neighbors (as edge ranks remain constant). The fastest implementation of the
neighborhood data structure therefore has for each vertex $v$ an unsorted
dynamic size array of neighbors and a dynamic size array of deleted neighbors.
When accessing a vertex $v$, we sort both arrays and remove the neighbors that
were since removed; and then clear the array of deleted neighbors. While this
incurs $\Oh{\Delta \log{\Delta}}$ running time, it is faster than both
balancing and non-balancing binary trees.

\subsubsection{Random Update Sequence on Generated Graph.}
\label{exp:m-random}

We performed experiments for the fully dynamic maximal matching problem using the same
random update sequences that we used in Section~\ref{exp:c-random} for the fully
dynamic $(\Delta + 1)$ coloring problem.

\begin{figure}[t!]
    \includegraphics[width=.49\textwidth]{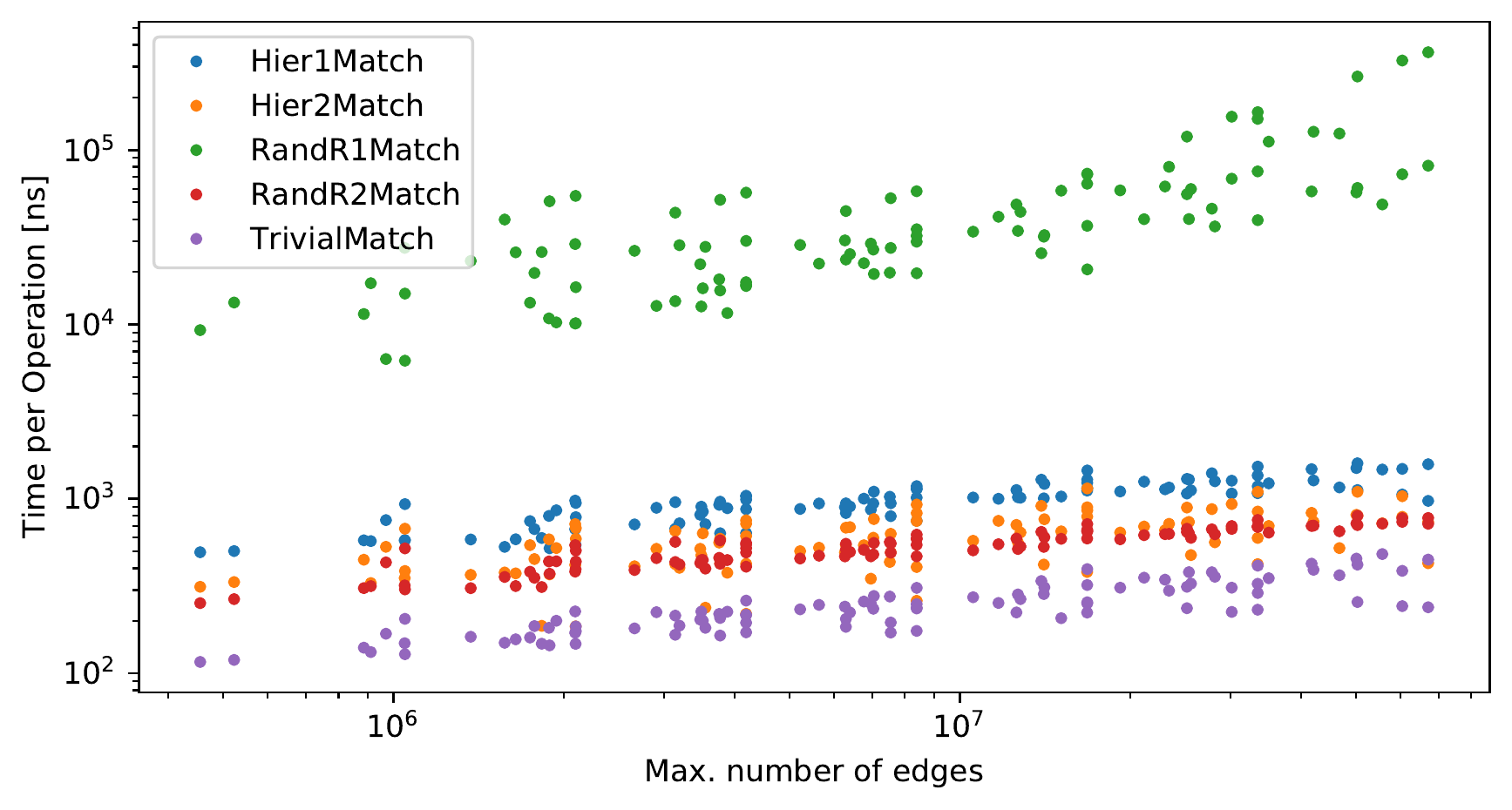}
    \caption{\label{fig:m-scaling} All random insertion and deletion sequences}
\end{figure}

In Figure~\ref{fig:m-scaling} we can see the results for this experiment. We can
clearly see that \matchbbhss~is more than an order slower than the other
algorithms on most instances and also scales the worst. The algorithm performs
especially poorly on instances with a high average degree. Most of the update time is
spent on updating the neighborhood data structure $N(v)$, as whenever an edge is
added or removed from the matching, the eliminator rank of all incident edges
needs to be updated, which results in a large number of updates in a
self-balancing binary search tree (implemented as \texttt{std::set}).
Similar to the fully-dynamic $(\Delta + 1)$ coloring problem, we can see that
the simple trivial algorithm \matchtriv~performs very well with random
insertions and deletions. For the closely related algorithms
\matchbgs~and~\matchsolomon~we can see that \matchsolomon~is faster all instances, with a
speedup factor of up to $2.5$, not significantly affected by
graph type, density, deletion rate or number of vertices. On most instances,
\matchrr~is~$10-30\%$ faster than \matchsolomon~,even though the algorithm does not
have good theoretical performance guarantees.
Looking at performance profiles we can see that \matchsolomon~spends up to
$40\%$ in hierarchy level updates, while the simpler two-level implementation of
\matchbgs~only spends $20-25\%$ in level updates. This is not all to surprising,
as the implementation in \matchsolomon~uses $\log{n}$ levels while the
implementation in \matchbgs~only uses two. For both algorithms this is the
largest part of the time spent in updating the maximal matching, the rest of the
running time of the algorithms is spent in adding and removing elements from
neighborhood data structures.

For \matchbbhss, the vast majority of running time is spent in updating the
neighborhood binary trees. \matchrr~spends about a quarter of the total running
time in Algorithm~\ref{alg:partner}, which updates the matching and
\matchtriv~spends about $10\%$ of the running time in matching updates. The rest of the respective running times are spent in updating
neighborhood data structures when inserting and removing edges.

\textbf{Size of the maximal matchings} We also evaluated the {size of the
maximal matchings found}. The geometric mean of the matching size
\matchbbhss~and~\matchrr~about $0.3\%$ lower than the geometric mean of the
matching size of the other algorithms whose matching sizes are all very similar.
This is not very surprising as the matchings of \matchbbhss~and~\matchrr~are
lexicographically first maximal matching, i.e., maximal matchings that must
fulfill an additional requirement.

\begin{figure}[t!]
    \centering
    \includegraphics[width=.4\textwidth]{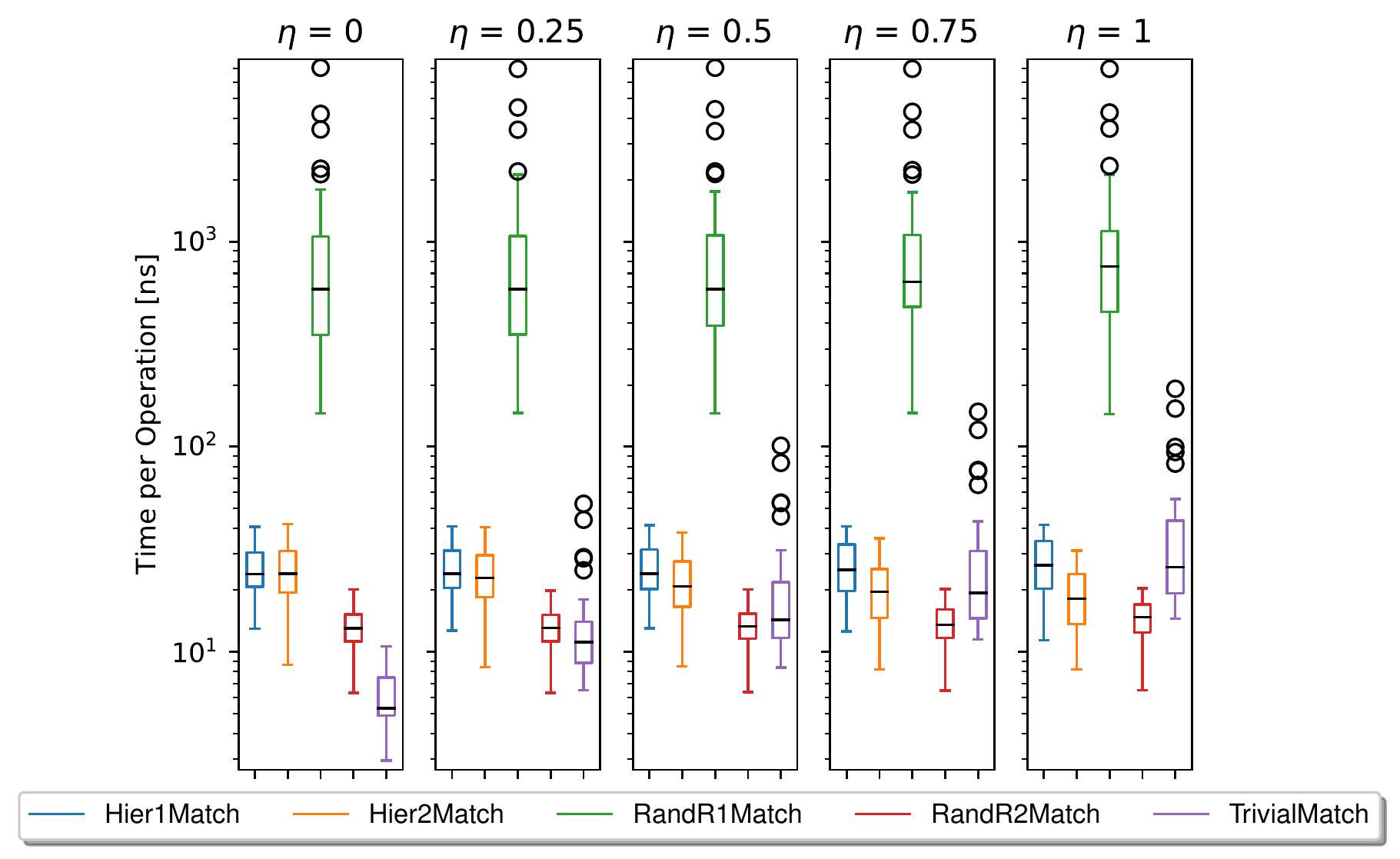}
    \caption{\label{fig:m-windows} Sliding window instances}
\end{figure}

\subsubsection{Generated Sliding Window Instances.}
\label{exp:m-windows}
In this experiment, we generate sliding window instances, in which we insert
edges from generated ER and RHG graphs in random order and delete the oldest
edge whenever the dynamic graph has $\phi$ (see below) edges. We also give an additional
parameter $\eta \in [0,1]$ that is used to generate harder instances for
\matchtriv. For each edge deletion, with probability $\eta$ we delete the oldest
edge that is part of the matching of \matchtriv, with probability $1 - \eta$ we
delete the oldest edge in the sliding window. Similar to the coloring problems,
an adversary that is oblivious to the random choices made by the algorithms cannot create such
worst-case instances for the non-trivial algorithms.
For either graph generator, we generate graphs with $n = 2^{14}, 2^{15}, 2^{16},
2^{17}, 2^{18}$ and $m = 2^{25}$. For each graph, we generated sliding windows with window size $\phi =
2^{20}, 2^{21}, 2^{22}, 2^{23}$ and $\eta = 0, 25\%, 50\%, 75\%, 100\%$, which
results in a total of $100$ instances.
Figure~\ref{fig:m-windows} shows box plots for these instances. In these box plots, the box extends from the lower to the upper quartile of the data, the line indicates the median. The whiskers extend to show the range, with outlier results denoted by black dots. The results with
$\eta = 0$ are similar to the results in Section~\ref{exp:m-random}, with
increased parameter $\eta$ \matchtriv~is much slower, as each deleted matching
edge means that both adjacent vertices are now unmatched and need to check
whether they have an incident edge to another unmatched vertex. However, even
when $\eta = 1$, \matchtriv~is still significantly faster than \matchbbhss.

\subsubsection{Real-world Graphs.}
\label{exp:m-wiki}

\begin{figure}[t!]
    \includegraphics[width=.49\textwidth]{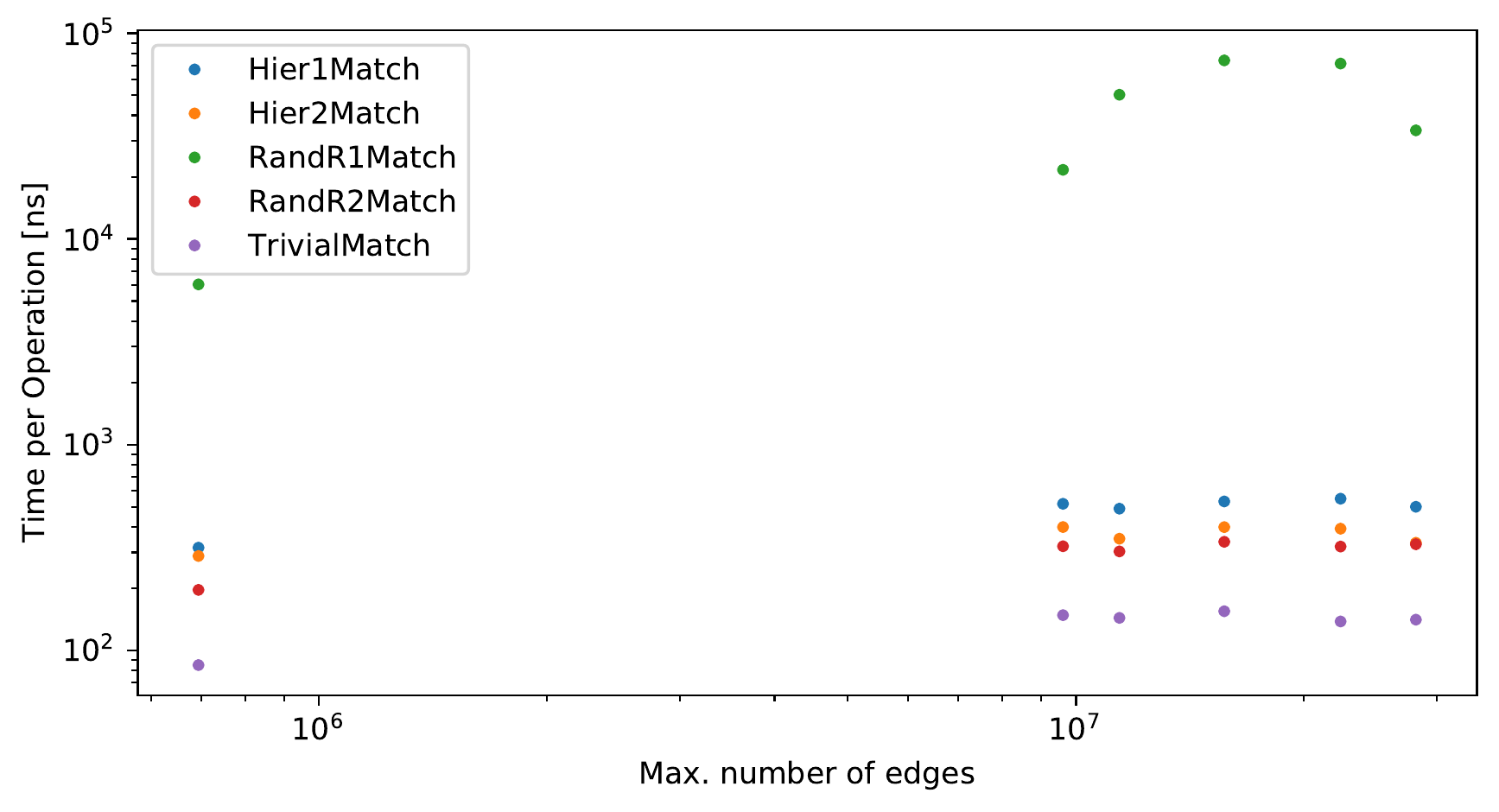}
    \caption{\label{fig:m-wiki} Real-world wikipedia edit instances}
\end{figure}

Figure~\ref{fig:m-wiki} shows results of the wikipedia edit
instances described in Section~\ref{exp:instances}. On these, the
algorithms have similar performances to the instances in
Section~\ref{exp:m-random}: \matchtriv~is the fastest algorithm, followed by
\matchrr~and~\matchsolomon, then followed by \matchbgs. The algorithm \matchbbhss~is slower by one to two orders of magnitude.

\section{Conclusion}\label{s:conclusion}

We empirically evaluated a variety of algorithms for the fully dynamic $(\Delta + 1)$ coloring
problem and the fully dynamic maximal matching problem. 
For both problems the trivial algorithms performed best, except on specifically chosen worst-case input sequences. 
The reason is the use of different data structures for storing the neighborhood of each vertex and how they interact with the caching system.
Dynamic algorithms with random rank also clearly outperformed 
algorithms using a hierarchical graph decompositions. The reason for this results is the frequent use of expensive data structures for updating the hierarchy.
However, for maximal matching the random-rank based algorithm with best asymptotic performance that time $O(\log^4 n)$ and it was clearly outperformed by the hierarchical algorithms with
$O(\log n)$ or constant expected update time. It would be very interesting to find a constant-time random rank algorithm and compare its performance empirically. We gave a random-rank based algorithm
that outperforms the constant-time hierarchical algorithms, but its asymptotic bound is $O(nm)$.

To summarize this work shows that the extensive use of data structures that 
support random accesses significantly slow down the non-trivial dynamic algorithm in practice, pointing to a (known) weakness of the RAM model which need to be adapted
to increase the impact of theoretical algorithms research. One such attempt is the external-memory model, but as it completely ignores computation time, it is also overlooking a large contribution to the running time of dynamic algorithms.

\bibliographystyle{abbrv} 
\bibliography{paper}

\clearpage

\begin{appendices}
\section{Additional Proofs}
\label{app:proofs}
\begin{customlemma}{4.1}  A matching $\match$ is $LFMM(G,\pi)$ if and
    only if for each edge $e = (u,v) \in E$ either $e \in \match$ or $\exists w \in
    N(u) \cup N(v)$ s.t. $(u,w) \in \match$ with $\pi(u,w) < \pi(e)$ or $(v,w) \in
    M$ with $\pi(v,w) < \pi(e)$.
    \end{customlemma}
    
    \begin{proof}
    In the construction algorithm of $LFMM(G, \pi)$, every edge is either added to
    the matching or covered by an incident edge. For all edges $e \not \in M$, there
    therefore exists an edge of smaller rank that is incident to $e$.
    
    To prove the other direction, let $\match$ be the $LFMM(G, \pi)$ and $\match_o$
    be a matching with $\match \neq \match_o$, i.e. they differ in at least one
    edge. As $\match$ is a maximal matching, $\match \not\subseteq \match_o$, i.e.
    it is not fully encompassed in matching $\match_o$. Thus, there exists an edge
    $e \in E$ that is in one of the matchings but not the other. Let $e_f \in E$ be
    the lowest rank edge that is in one matching but not the other. For every edge
    $e_l \in E$ with $\pi(e_l) < \pi(e_f)$, $e_l \in M$ if and only if $e_l \in
    \match_o$. We now look at both cases where $e_f \in M$ and $e_f \in M_o$ and
    show that $\match_o$ is not an LFMM for $G$ with ranking $\pi$.
    
    \emph{Case 1: $e_f \in \match_o$, $e_f \not\in \match$.} As $e_f \not\in \match$, it must
    have been covered by an incident edge of smaller rank. Thus there exists an edge
    $e_x \in M$ incident to $e_f$ with $\pi(e_x) < \pi(e_f)$ that covers $e_f$. Due
    to the assumption that $e_f$ is the lowest rank edge in which the matchings
    differ, $e_x \in \match_o$ and $\match_o$ is not a matching.
    
    \emph{Case 2: $e_f \not\in \match_o$, $e_f \in \match$.} As $e_f \in \match$, we know
    that no incident edge of rank $< e_f$ covers $e_f$. As the matchings do not
    differ in any edge of rank $< e_f$, this is also true for $\match_o$. Therefore
    $e_f \not\in M_o$ but no incident matching edge with smaller rank exists. This
    contradicts the assumption that every non-matching edge is covered.   
    
    This shows that any matching that differs from $\match$ in at least one edge is
    not $LFMM(G,\pi)$.
    \end{proof}

    \begin{customclaim}{1} 
        In Algorithm~\ref{alg:partner}, when vertex $v$ is removed from $S$ with
        rank $r_v$, then any new vertex $w$ that is added to $S$ during this
        iteration of Algorithm~\ref{alg:partner} is added with rank $> r_v$.
    \end{customclaim}

    \begin{proof}
        In Algorithm~\ref{alg:partner}, we try to find a incident matching edge for
        vertex $v$. We iterate over all incident edges $e = (v,w)$ with $\pi(e) >
        r_v$. The edge $e$ is added to the matching, if $\pi(e) < k(v)$ and $\pi(e)
        < k(w)$ and we add the current matching partners of $v$ and $w$ to $S$ (if they
        exist). $P(v)$ (and $P(w)$ analogously) is added to $S$ with priority
        $k(v)$, which is $> \pi(v) > r_v$.
    \end{proof}
    
    \begin{customclaim}{2}
        Algorithm~\ref{alg:partner} checks vertices in non-decreasing priority
        order.
    \end{customclaim}
    
    \begin{proof}
        Let $v$ and $r_v$ be the vertex and its priority that is currently checked
        in the loop of Algorithm~\ref{alg:partner}. In order to show that all
        vertices checked at a later point of the algorithm have priority $\geq
        r_v$, we need to show the following: (1) All vertices that are currently in
        $S$ have priority $\geq r_v$ and (2) all vertices that are added to $S$ at
        any later point have priority $\geq r_v$. (1) follows from the fact that $S$
        is a priority queue and returns the vertex with smallest priority, (2) follows directly from
        Claim~\ref{claim:onlylarger}.
    \end{proof}

    \begin{customlemma}{4.3} Given the matching $\match =
        LFMM(G,\pi)$ and the insertion of an edge $e = (u,v)$,
        Algorithm~\ref{alg:insert} maintains the lexicographically first maximum
        matching $LFMM(G+e,\pi)$.
        \end{customlemma}
        
        \begin{proof}
        Algorithm~\ref{alg:insert} adds $e$ to the data structures and checks whether it
        is uncovered, i.e. has to be inserted into matching $\match$. In this case, we
        need to evict the current matching edges incident to $u$ and $v$ from $\match$
        (if they exist). All edges incident to $u$ and $v$ are still covered, as the
        ranking of their incident matching edge only decreased, but their previous
        matching partners $P(u)$ and $P(v)$ are now unmatched. Let $e_x = (u, P(u))$ be
        an evicted matching edge (analogous for $v$ and $P(v)$). 
        Note that $e_x$ is covered by $e$.
        However, all edges $e_2$
        incident to $P(u)$ with $e_2 > e_x$ were covered by $e_x$ and are now
        potentially uncovered. Edges $e_2$ incident to $P(u)$ with $e_x < e_2$ were
        previously already not covered by $e_x$ and therefore are guaranteed to remain
        covered even though $e_x$ was removed from the matching. As no other edges were changed, we know that all edges not incident to
        $P(u)$ and $P(v)$ and all edges of rank $\le$min$(\pi(u, P(u)), \pi(v, P(v)))$ are
        covered. We add the unmatched partner vertices $P(u)$ and $P(v)$ (if they exist)
        to priority queue $S$. The invariants I1 and I2 currently hold as only edges incident to
        $P(u)$ and $P(v)$ with rank greater than their original matching edge are
        potentially uncovered.
        
        Then Algorithm~\ref{alg:partner} is called and the correctness follows from Lemma~\ref{lem:partner}.
        \end{proof}
        
        \begin{customlemma}{4.4} Given the matching $\match =
             LFMM(G,\pi)$ and the deletion of an edge $e = (u,v)$,
             Algorithm~\ref{alg:delete} maintains the lexicographically first maximum
             matching $LFMM(G-e,\pi)$.
        \end{customlemma}
        
        \begin{proof}
         In Algorithm~\ref{alg:delete}, we delete an edge $e$ and check whether it is in
         the matching $\match$. If it is not, no edges were covered by $e$, thus every
         edge in $G$ is still covered after removal of $e$, and $\match$ is still the
         $LFMM(G-e,\pi)$ due to Lemma~\ref{lem:killed}. If $e$ is a matching edge, we
         remove it from $\match$ and add $u$ and $v$ to $S$, as all incident edges with
         rank $> \pi(e)$ are now potentially uncovered. Analogously to the insertion in
         Lemma~\ref{lem:insert-correct}, all potentially uncovered edges are incident to
         $u$ or $v$ and have rank $> \pi(e)$. Thus Invariants I1 and I2 hold.
         Then Algorithm~\ref{alg:partner} is called and the correctness follows from Lemma~\ref{lem:partner}.
        \end{proof}

        \begin{customclaim}{3}
            Algorithm~\ref{alg:partner} terminates after at most $m$ iterations. 
        \end{customclaim}
        
        \begin{proof}
            In order to prove that the update operations terminate, we need to look at
            how vertices are added to the priority queue $S$, as the algorithm
            terminates as soon as $S$ is empty. A vertex $v$ is added with priority
            $r_v$ exactly when an incident edge $e$ with rank $r_v$ is removed from the
            matching $\match$. We process the vertices in $S$ in increasing priority
            order, thus, according to Claim~\ref{claim:non-decreasing}, when $v$ with
            priority $r_v$ is at the front of $S$, all edges with rank $< r_v$ are
            \emph{settled}, i.e. they will not be added to or removed from $\match$ at a
            later point of the update operation. In Algorithm~\ref{alg:partner}, we try
            to find a new matching edge $e'$ and check all incident edges with
            rank $> r_v$. Thus, as the edge ranks are all unique, $e'$ is now settled as well and will not be updated
            at a later point. As each iteration of the while-loop in
            Algorithm~\ref{alg:partner} settles at least one edge, the algorithm
            terminates after at most $m$ iterations.
        \end{proof}

\section{Additional Pseudocode and Tables}
\label{app:pseudocode}

\begin{algorithm}[h!]
    \caption{Edge Insertion \label{alg:insert}}
    \begin{algorithmic}[1]
        \INPUT $e = (u,v) \leftarrow $: edge to insert
        \State $\pi(e) \leftarrow$ rand$(0,1)$
        \State $N(v)$.insert$(\pi(e), u)$
        \State $N(u)$.insert$(\pi(e), v)$
        \State $S \leftarrow $ empty priority queue
        \If{$\pi(e) < min(k(v), k(u))$}
        \State $p_v \leftarrow P(v)$
        \If{$p_v \neq -1$}
            \State S.insert($p_v$, $k(p_v)$)
            \State $k(p_v) \leftarrow \infty$
            \State $P(p_v) \leftarrow -1$ 
        \EndIf
        \State $p_u \leftarrow P(u)$
        \If{$p_u \neq -1$}
            \State S.insert($p_u$, $k(p_u)$)
            \State $k(p_u) \leftarrow \infty$
            \State $P(p_u) \leftarrow -1$ 
        \EndIf
        \State $k(u) \leftarrow \pi(e)$ 
        \State $k(v) \leftarrow \pi(e)$
        \State $P(u) \leftarrow v$
        \State $P(v) \leftarrow u$
        \EndIf
        \State findNewPartners($S$)
    \end{algorithmic}
\end{algorithm}

\begin{algorithm}[h!]
    \caption{Edge Deletion \label{alg:delete}}
    \begin{algorithmic}[1]
        \INPUT $e = (u,v) \leftarrow$ edge to remove
        \State $N(v)$.delete$(u)$
        \State $N(u)$.delete$(v)$
        \State $S \leftarrow $ empty priority queue
        \If{$P(v) = u$}
            \State S.insert($u, \pi(e)$)
            \State S.insert($v, \pi(e)$)
            \State $k(v) \leftarrow \infty$
            \State $k(u) \leftarrow \infty$
            \State $P(v) \leftarrow -1$
            \State $P(u) \leftarrow -1$
        \EndIf

        \State findNewPartners($S$)
    \end{algorithmic}
\end{algorithm}

\begin{table*}[!h]
    \centering
\begin{tabular}{l|r|r|r|r}
    $\rho$ & \colbgkls & \colrec & \colhp & \colcount \\ \hline \hline
    0.0 & \numprint{1385.65} (\numprint{9.39}) & \numprint{147.58} (\numprint{1.00}) & \numprint{838.87} (\numprint{5.68}) & \numprint{440.00} (\numprint{2.98})  \\
    0.1 & \numprint{1366.87} (\numprint{5.94}) & \numprint{230.04} (\numprint{1.00}) & \numprint{775.95} (\numprint{3.37}) & \numprint{548.09} (\numprint{2.38})  \\
    0.25 & \numprint{1314.98} (\numprint{4.42}) & \numprint{297.32} (\numprint{1.00}) & \numprint{780.06} (\numprint{2.62}) & \numprint{623.54} (\numprint{2.10})  \\
    0.5 & \numprint{1208.33} (\numprint{3.83}) & \numprint{315.64} (\numprint{1.00}) & \numprint{682.45} (\numprint{2.16}) & \numprint{634.96} (\numprint{2.01})  \\
    0.75 & \numprint{1106.47} (\numprint{4.40}) & \numprint{251.55} (\numprint{1.00}) & \numprint{615.94} (\numprint{2.45}) & \numprint{528.27} (\numprint{2.10})  \\

\end{tabular}
\caption{Average time in $ns$ per operation (and slowdown to fastest), Section~\ref{exp:c-random}, for all different values of $\rho$.}
\end{table*}

\begin{table*}[!h]
    \centering
    \begin{tabular}{l|r|r|r|r}
        Graph & \colbgkls & \colrec & \colhp & \colcount \\ \hline \hline
ER $n=2^{13}, m = n \cdot 2^{8}$ & \numprint{612.57} (\numprint{5.61}) & \numprint{109.16} (\numprint{1.00}) & \numprint{317.38} (\numprint{2.91}) & \numprint{197.19} (\numprint{1.81})  \\
ER $n=2^{14}, m = n \cdot 2^{8}$ & \numprint{733.53} (\numprint{5.52}) & \numprint{132.92} (\numprint{1.00}) & \numprint{373.52} (\numprint{2.81}) & \numprint{240.90} (\numprint{1.81})  \\
ER $n=2^{15}, m = n \cdot 2^{8}$ & \numprint{974.99} (\numprint{6.12}) & \numprint{159.26} (\numprint{1.00}) & \numprint{500.24} (\numprint{3.14}) & \numprint{312.55} (\numprint{1.96})  \\
ER $n=2^{16}, m = n \cdot 2^{6}$ & \numprint{1165.88} (\numprint{7.47}) & \numprint{156.14} (\numprint{1.00}) & \numprint{679.39} (\numprint{4.35}) & \numprint{400.40} (\numprint{2.56})  \\
ER $n=2^{16}, m = n \cdot 2^{7}$ & \numprint{1252.77} (\numprint{7.11}) & \numprint{176.28} (\numprint{1.00}) & \numprint{708.24} (\numprint{4.02}) & \numprint{433.68} (\numprint{2.46})  \\
ER $n=2^{16}, m = n \cdot 2^{8}$ & \numprint{1371.70} (\numprint{6.46}) & \numprint{212.18} (\numprint{1.00}) & \numprint{745.25} (\numprint{3.51}) & \numprint{505.44} (\numprint{2.38})  \\
ER $n=2^{16}, m = n \cdot 2^{9}$ & \numprint{1479.21} (\numprint{5.57}) & \numprint{265.45} (\numprint{1.00}) & \numprint{800.28} (\numprint{3.01}) & \numprint{584.39} (\numprint{2.20})  \\
ER $n=2^{16}, m = n \cdot 2^{10}$ & \numprint{1645.04} (\numprint{4.53}) & \numprint{362.78} (\numprint{1.00}) & \numprint{886.99} (\numprint{2.45}) & \numprint{750.66} (\numprint{2.07})  \\
ER $n=2^{17}, m = n \cdot 2^{8}$ & \numprint{1742.95} (\numprint{6.08}) & \numprint{286.83} (\numprint{1.00}) & \numprint{1001.00} (\numprint{3.49}) & \numprint{701.61} (\numprint{2.45})  \\
ER $n=2^{18}, m = n \cdot 2^{8}$ & \numprint{2111.15} (\numprint{5.53}) & \numprint{381.66} (\numprint{1.00}) & \numprint{1260.63} (\numprint{3.30}) & \numprint{935.03} (\numprint{2.45})  \\
RHG $n=2^{13}, m = n \cdot 2^{8}$ & \numprint{506.94} (\numprint{4.40}) & \numprint{115.24} (\numprint{1.00}) & \numprint{297.43} (\numprint{2.58}) & \numprint{223.28} (\numprint{1.94})  \\
RHG $n=2^{14}, m = n \cdot 2^{8}$ & \numprint{673.94} (\numprint{4.65}) & \numprint{144.82} (\numprint{1.00}) & \numprint{384.58} (\numprint{2.66}) & \numprint{286.35} (\numprint{1.98})  \\
RHG $n=2^{15}, m = n \cdot 2^{8}$ & \numprint{985.15} (\numprint{5.18}) & \numprint{190.01} (\numprint{1.00}) & \numprint{515.13} (\numprint{2.71}) & \numprint{421.89} (\numprint{2.22})  \\
RHG $n=2^{16}, m = n \cdot 2^{6}$ & \numprint{1158.60} (\numprint{7.20}) & \numprint{160.91} (\numprint{1.00}) & \numprint{618.36} (\numprint{3.84}) & \numprint{465.90} (\numprint{2.90})  \\
RHG $n=2^{16}, m = n \cdot 2^{7}$ & \numprint{1250.88} (\numprint{5.92}) & \numprint{211.21} (\numprint{1.00}) & \numprint{674.32} (\numprint{3.19}) & \numprint{545.16} (\numprint{2.58})  \\
RHG $n=2^{16}, m = n \cdot 2^{8}$ & \numprint{1278.62} (\numprint{4.85}) & \numprint{263.53} (\numprint{1.00}) & \numprint{680.36} (\numprint{2.58}) & \numprint{607.43} (\numprint{2.30})  \\
RHG $n=2^{16}, m = n \cdot 2^{9}$ & \numprint{1408.32} (\numprint{4.12}) & \numprint{341.49} (\numprint{1.00}) & \numprint{727.64} (\numprint{2.13}) & \numprint{704.60} (\numprint{2.06})  \\
RHG $n=2^{16}, m = n \cdot 2^{10}$ & \numprint{1520.41} (\numprint{3.29}) & \numprint{462.34} (\numprint{1.00}) & \numprint{824.23} (\numprint{1.78}) & \numprint{863.87} (\numprint{1.87})  \\
RHG $n=2^{17}, m = n \cdot 2^{8}$ & \numprint{1636.89} (\numprint{4.66}) & \numprint{351.57} (\numprint{1.00}) & \numprint{905.24} (\numprint{2.57}) & \numprint{817.54} (\numprint{2.33})  \\
RHG $n=2^{18}, m = n \cdot 2^{8}$ & \numprint{2019.66} (\numprint{4.17}) & \numprint{484.80} (\numprint{1.00}) & \numprint{1858.53} (\numprint{3.83}) & \numprint{1101.60} (\numprint{2.27})  \\

\end{tabular}
\caption{Average time in $ns$ per operation (and slowdown to fastest), Section~\ref{exp:c-random}, for different Erd\"os-Renyi and random hyperbolic geometric graph families }
\end{table*}

\begin{table*}[!h]
    \centering
    \begin{tabular}{l|r|r|r|r}
        & \colbgkls & \colrec & \colhp & \colcount \\ \hline \hline
$n=2^{12}, m = n \cdot 2^{7}$ & \numprint{384.92} (\numprint{4.25}) & \numprint{200.51} (\numprint{2.21}) & \numprint{237.26} (\numprint{2.62}) & \numprint{90.66} (\numprint{1.00})  \\
$n=2^{12}, m = n \cdot 2^{8}$ & \numprint{484.63} (\numprint{4.05}) & \numprint{308.30} (\numprint{2.57}) & \numprint{280.07} (\numprint{2.34}) & \numprint{119.73} (\numprint{1.00})  \\
$n=2^{12}, m = n \cdot 2^{9}$ & \numprint{632.33} (\numprint{4.28}) & \numprint{773.94} (\numprint{5.24}) & \numprint{320.88} (\numprint{2.17}) & \numprint{147.61} (\numprint{1.00})  \\
$n=2^{12}, m = n \cdot 2^{10}$ & \numprint{712.27} (\numprint{4.21}) & \numprint{1959.38} (\numprint{11.57}) & \numprint{360.09} (\numprint{2.13}) & \numprint{169.29} (\numprint{1.00})  \\
$n=2^{13}, m = n \cdot 2^{7}$ & \numprint{497.63} (\numprint{3.89}) & \numprint{226.20} (\numprint{1.77}) & \numprint{302.26} (\numprint{2.36}) & \numprint{128.02} (\numprint{1.00})  \\
$n=2^{13}, m = n \cdot 2^{8}$ & \numprint{583.08} (\numprint{3.63}) & \numprint{455.31} (\numprint{2.83}) & \numprint{345.00} (\numprint{2.15}) & \numprint{160.83} (\numprint{1.00})  \\
$n=2^{13}, m = n \cdot 2^{9}$ & \numprint{723.19} (\numprint{4.09}) & \numprint{1049.81} (\numprint{5.94}) & \numprint{373.34} (\numprint{2.11}) & \numprint{176.64} (\numprint{1.00})  \\
$n=2^{13}, m = n \cdot 2^{10}$ & \numprint{811.83} (\numprint{4.28}) & \numprint{2297.92} (\numprint{12.12}) & \numprint{407.48} (\numprint{2.15}) & \numprint{189.58} (\numprint{1.00})  \\
$n=2^{14}, m = n \cdot 2^{7}$ & \numprint{659.06} (\numprint{3.68}) & \numprint{328.14} (\numprint{1.83}) & \numprint{395.15} (\numprint{2.21}) & \numprint{179.02} (\numprint{1.00})  \\
$n=2^{14}, m = n \cdot 2^{8}$ & \numprint{770.18} (\numprint{3.99}) & \numprint{608.97} (\numprint{3.16}) & \numprint{424.16} (\numprint{2.20}) & \numprint{192.90} (\numprint{1.00})  \\
$n=2^{14}, m = n \cdot 2^{9}$ & \numprint{868.08} (\numprint{4.38}) & \numprint{1144.72} (\numprint{5.77}) & \numprint{436.99} (\numprint{2.20}) & \numprint{198.23} (\numprint{1.00})  \\
$n=2^{14}, m = n \cdot 2^{10}$ & \numprint{955.15} (\numprint{4.68}) & \numprint{2539.51} (\numprint{12.43}) & \numprint{464.66} (\numprint{2.27}) & \numprint{204.27} (\numprint{1.00})  \\
$n=2^{15}, m = n \cdot 2^{7}$ & \numprint{872.28} (\numprint{3.61}) & \numprint{425.36} (\numprint{1.76}) & \numprint{539.78} (\numprint{2.23}) & \numprint{241.87} (\numprint{1.00})  \\
$n=2^{15}, m = n \cdot 2^{8}$ & \numprint{965.36} (\numprint{3.83}) & \numprint{780.35} (\numprint{3.10}) & \numprint{565.28} (\numprint{2.24}) & \numprint{251.98} (\numprint{1.00})  \\
$n=2^{15}, m = n \cdot 2^{9}$ & \numprint{1107.30} (\numprint{4.31}) & \numprint{1491.58} (\numprint{5.81}) & \numprint{574.00} (\numprint{2.24}) & \numprint{256.68} (\numprint{1.00})  \\
$n=2^{15}, m = n \cdot 2^{10}$ & \numprint{1223.99} (\numprint{4.69}) & \numprint{2911.95} (\numprint{11.15}) & \numprint{640.29} (\numprint{2.45}) & \numprint{261.22} (\numprint{1.00})  \\
        
\end{tabular}
\caption{Average time in $ns$ per operation (and slowdown to fastest), Section~\ref{exp:c-clash}, Clashing Sequence for \colrec}
\end{table*}

\begin{table*}[!h]
    \centering
    \begin{tabular}{l|r|r|r|r}
        & \colbgkls & \colrec & \colhp & \colcount \\ \hline \hline
$n=2^{12}, m = n \cdot 2^{7}$ & \numprint{438.18} (\numprint{6.87}) & \numprint{63.76} (\numprint{1.00}) & \numprint{250.12} (\numprint{3.92}) & \numprint{412.22} (\numprint{6.47})  \\
$n=2^{12}, m = n \cdot 2^{8}$ & \numprint{536.57} (\numprint{7.88}) & \numprint{68.07} (\numprint{1.00}) & \numprint{291.50} (\numprint{4.28}) & \numprint{1092.24} (\numprint{16.05})  \\
$n=2^{12}, m = n \cdot 2^{9}$ & \numprint{639.16} (\numprint{8.21}) & \numprint{77.83} (\numprint{1.00}) & \numprint{341.07} (\numprint{4.38}) & \numprint{4693.29} (\numprint{60.30})  \\
$n=2^{12}, m = n \cdot 2^{10}$ & \numprint{782.89} (\numprint{8.57}) & \numprint{91.38} (\numprint{1.00}) & \numprint{360.64} (\numprint{3.95}) & \numprint{17327.34} (\numprint{189.62})  \\
$n=2^{13}, m = n \cdot 2^{7}$ & \numprint{543.10} (\numprint{6.98}) & \numprint{77.81} (\numprint{1.00}) & \numprint{311.17} (\numprint{4.00}) & \numprint{601.15} (\numprint{7.73})  \\
$n=2^{13}, m = n \cdot 2^{8}$ & \numprint{632.55} (\numprint{7.29}) & \numprint{86.83} (\numprint{1.00}) & \numprint{351.94} (\numprint{4.05}) & \numprint{2148.88} (\numprint{24.75})  \\
$n=2^{13}, m = n \cdot 2^{9}$ & \numprint{751.01} (\numprint{7.66}) & \numprint{98.02} (\numprint{1.00}) & \numprint{464.69} (\numprint{4.74}) & \numprint{7897.90} (\numprint{80.57})  \\
$n=2^{13}, m = n \cdot 2^{10}$ & \numprint{844.29} (\numprint{8.16}) & \numprint{103.42} (\numprint{1.00}) & \numprint{406.32} (\numprint{3.93}) & \numprint{22930.45} (\numprint{221.72})  \\
$n=2^{14}, m = n \cdot 2^{7}$ & \numprint{675.75} (\numprint{6.88}) & \numprint{98.18} (\numprint{1.00}) & \numprint{387.18} (\numprint{3.94}) & \numprint{1137.84} (\numprint{11.59})  \\
$n=2^{14}, m = n \cdot 2^{8}$ & \numprint{769.93} (\numprint{7.29}) & \numprint{105.62} (\numprint{1.00}) & \numprint{435.88} (\numprint{4.13}) & \numprint{3706.01} (\numprint{35.09})  \\
$n=2^{14}, m = n \cdot 2^{9}$ & \numprint{903.79} (\numprint{8.30}) & \numprint{108.85} (\numprint{1.00}) & \numprint{448.75} (\numprint{4.12}) & \numprint{10364.31} (\numprint{95.21})  \\
$n=2^{14}, m = n \cdot 2^{10}$ & \numprint{959.41} (\numprint{8.69}) & \numprint{110.40} (\numprint{1.00}) & \numprint{462.33} (\numprint{4.19}) & \numprint{27090.15} (\numprint{245.39})  \\
$n=2^{15}, m = n \cdot 2^{7}$ & \numprint{901.02} (\numprint{7.66}) & \numprint{117.55} (\numprint{1.00}) & \numprint{556.56} (\numprint{4.73}) & \numprint{1978.43} (\numprint{16.83})  \\
$n=2^{15}, m = n \cdot 2^{8}$ & \numprint{1016.79} (\numprint{8.68}) & \numprint{117.11} (\numprint{1.00}) & \numprint{577.63} (\numprint{4.93}) & \numprint{5462.71} (\numprint{46.65})  \\
$n=2^{15}, m = n \cdot 2^{9}$ & \numprint{1141.97} (\numprint{9.71}) & \numprint{117.58} (\numprint{1.00}) & \numprint{598.01} (\numprint{5.09}) & \numprint{14135.98} (\numprint{120.22})  \\
$n=2^{15}, m = n \cdot 2^{10}$ & \numprint{1262.60} (\numprint{10.75}) & \numprint{117.48} (\numprint{1.00}) & \numprint{644.49} (\numprint{5.49}) & \numprint{34874.17} (\numprint{296.85})  \\

\end{tabular}
\caption{Average time in $ns$ per operation (and slowdown to fastest), Section~\ref{exp:c-clash}, Clashing Sequence for \colcount}
\end{table*}

\begin{table*}[!h]
    \centering
    \begin{tabular}{l|r|r|r|r|r}
        Graph & Avg deg. / max deg. & \colbgkls & \colrec & \colhp & \colcount \\ \hline \hline
        $n=2^{13}, m= n \cdot 2^{9}$ & 0.99283 & \numprint{2333.08} (\numprint{3.68}) & \numprint{692.55} (\numprint{1.09}) & \numprint{1386.77} (\numprint{2.19}) & \numprint{634.04} (\numprint{1.00})  \\
        $n=2^{13}, m= n \cdot 2^{10}$ & 0.98625 & \numprint{3114.35} (\numprint{3.91}) & \numprint{795.91} (\numprint{1.00}) & \numprint{1613.66} (\numprint{2.03}) & \numprint{877.70} (\numprint{1.10})  \\
        $n=2^{13}, m= n \cdot 2^{11}$ & 0.98437 & \numprint{3267.36} (\numprint{3.65}) & \numprint{895.82} (\numprint{1.00}) & \numprint{1816.83} (\numprint{2.03}) & \numprint{998.68} (\numprint{1.11})  \\
        $n=2^{13}, m= n \cdot 2^{12}$ & 0.98584 & \numprint{5101.24} (\numprint{3.27}) & \numprint{1560.94} (\numprint{1.00}) & \numprint{2659.12} (\numprint{1.70}) & \numprint{1675.18} (\numprint{1.07})  \\
        $n=2^{14}, m= n \cdot 2^{9}$ & 0.99814 & \numprint{2972.85} (\numprint{3.42}) & \numprint{1441.01} (\numprint{1.66}) & \numprint{1865.96} (\numprint{2.15}) & \numprint{869.67} (\numprint{1.00})  \\
$n=2^{14}, m= n \cdot 2^{10}$ & 0.99423 & \numprint{3374.15} (\numprint{3.60}) & \numprint{1028.43} (\numprint{1.10}) & \numprint{2004.28} (\numprint{2.14}) & \numprint{936.67} (\numprint{1.00})  \\
$n=2^{14}, m= n \cdot 2^{11}$ & 0.99114 & \numprint{3728.04} (\numprint{3.78}) & \numprint{985.55} (\numprint{1.00}) & \numprint{2133.60} (\numprint{2.16}) & \numprint{1058.80} (\numprint{1.07})  \\
$n=2^{14}, m= n \cdot 2^{12}$ & 0.98938 & \numprint{4663.57} (\numprint{3.47}) & \numprint{1344.03} (\numprint{1.00}) & \numprint{2625.03} (\numprint{1.95}) & \numprint{1566.44} (\numprint{1.17})  \\
$n=2^{15}, m= n \cdot 2^{9}$ & 0.99952 & \numprint{3732.84} (\numprint{3.22}) & \numprint{2073.84} (\numprint{1.79}) & \numprint{2251.59} (\numprint{1.94}) & \numprint{1159.05} (\numprint{1.00})  \\
$n=2^{15}, m= n \cdot 2^{10}$ & 0.99812 & \numprint{4128.14} (\numprint{3.32}) & \numprint{1335.69} (\numprint{1.07}) & \numprint{2452.47} (\numprint{1.97}) & \numprint{1242.67} (\numprint{1.00})  \\
$n=2^{15}, m= n \cdot 2^{11}$ & 0.99518 & \numprint{4503.01} (\numprint{3.66}) & \numprint{1232.00} (\numprint{1.00}) & \numprint{2672.62} (\numprint{2.17}) & \numprint{1492.68} (\numprint{1.21})  \\
$n=2^{15}, m= n \cdot 2^{12}$ & 0.99324 & \numprint{4736.47} (\numprint{3.85}) & \numprint{1231.70} (\numprint{1.00}) & \numprint{2856.47} (\numprint{2.32}) & \numprint{1660.23} (\numprint{1.35})  \\
$n=2^{16}, m= n \cdot 2^{9}$ & 0.99988 & \numprint{4679.67} (\numprint{2.69}) & \numprint{3200.96} (\numprint{1.84}) & \numprint{2830.52} (\numprint{1.63}) & \numprint{1738.15} (\numprint{1.00})  \\
$n=2^{16}, m= n \cdot 2^{10}$ & 0.99952 & \numprint{4621.44} (\numprint{2.43}) & \numprint{2529.88} (\numprint{1.33}) & \numprint{3071.81} (\numprint{1.62}) & \numprint{1899.85} (\numprint{1.00})  \\
$n=2^{16}, m= n \cdot 2^{11}$ & 0.99812 & \numprint{5105.18} (\numprint{2.65}) & \numprint{4341.48} (\numprint{2.25}) & \numprint{3371.23} (\numprint{1.75}) & \numprint{1929.09} (\numprint{1.00})  \\
$n=2^{16}, m= n \cdot 2^{12}$ & 0.99615 & \numprint{5329.71} (\numprint{2.50}) & \numprint{3021.18} (\numprint{1.42}) & \numprint{3496.29} (\numprint{1.64}) & \numprint{2129.03} (\numprint{1.00})  \\

\end{tabular}
\caption{Average time in $ns$ per operation (and slowdown to fastest), Section~\ref{exp:c-samedeg}}
\end{table*}

\begin{table*}[!h]
    \centering
\begin{tabular}{l|r|r|r|r|r}
    $\rho$ & \matchbbhss & \matchbgs & \matchrr & \matchsolomon & \matchtriv \\ \hline \hline
0.0 & \numprint{68232.92} (\numprint{254.55}) & \numprint{1010.38} (\numprint{3.77}) & \numprint{558.35} (\numprint{2.08}) & \numprint{417.68} (\numprint{1.56}) & \numprint{268.06} (\numprint{1.00})  \\
0.1 & \numprint{65916.21} (\numprint{254.60}) & \numprint{1005.68} (\numprint{3.88}) & \numprint{563.03} (\numprint{2.17}) & \numprint{620.09} (\numprint{2.40}) & \numprint{258.90} (\numprint{1.00})  \\
0.25 & \numprint{56402.74} (\numprint{219.98}) & \numprint{1011.97} (\numprint{3.95}) & \numprint{548.94} (\numprint{2.14}) & \numprint{633.58} (\numprint{2.47}) & \numprint{256.39} (\numprint{1.00})  \\
0.5 & \numprint{38825.19} (\numprint{158.09}) & \numprint{1002.12} (\numprint{4.08}) & \numprint{526.47} (\numprint{2.14}) & \numprint{658.69} (\numprint{2.68}) & \numprint{245.59} (\numprint{1.00})  \\
0.75 & \numprint{20261.76} (\numprint{94.55}) & \numprint{911.05} (\numprint{4.25}) & \numprint{483.74} (\numprint{2.26}) & \numprint{667.51} (\numprint{3.11}) & \numprint{214.29} (\numprint{1.00})  \\

\end{tabular}
\caption{Average time in $ns$ per operation (and slowdown to fastest), Section~\ref{exp:m-random}, for all different values of $\rho$.}
\end{table*}

\begin{table*}[!h]
    \centering
    \begin{tabular}{l|r|r|r|r|r}
        Graph & \matchbbhss & \matchbgs & \matchrr & \matchsolomon & \matchtriv \\ \hline \hline
        ER $n=2^{13}, m = n \cdot 2^{8}$ & \numprint{37194.02} (\numprint{261.91}) & \numprint{535.18} (\numprint{3.77}) & \numprint{341.80} (\numprint{2.41}) & \numprint{323.63} (\numprint{2.28}) & \numprint{142.01} (\numprint{1.00})  \\
        ER $n=2^{14}, m = n \cdot 2^{8}$ & \numprint{39242.31} (\numprint{244.56}) & \numprint{639.29} (\numprint{3.98}) & \numprint{390.64} (\numprint{2.43}) & \numprint{374.46} (\numprint{2.33}) & \numprint{160.46} (\numprint{1.00})  \\
        ER $n=2^{15}, m = n \cdot 2^{8}$ & \numprint{40378.22} (\numprint{224.28}) & \numprint{806.29} (\numprint{4.48}) & \numprint{480.36} (\numprint{2.67}) & \numprint{500.07} (\numprint{2.78}) & \numprint{180.03} (\numprint{1.00})  \\
        ER $n=2^{16}, m = n \cdot 2^{6}$ & \numprint{12425.32} (\numprint{55.74}) & \numprint{974.62} (\numprint{4.37}) & \numprint{553.41} (\numprint{2.48}) & \numprint{612.15} (\numprint{2.75}) & \numprint{222.90} (\numprint{1.00})  \\
        ER $n=2^{16}, m = n \cdot 2^{7}$ & \numprint{21547.75} (\numprint{103.62}) & \numprint{967.58} (\numprint{4.65}) & \numprint{536.05} (\numprint{2.58}) & \numprint{624.13} (\numprint{3.00}) & \numprint{207.94} (\numprint{1.00})  \\
        ER $n=2^{16}, m = n \cdot 2^{8}$ & \numprint{44139.36} (\numprint{197.16}) & \numprint{1088.99} (\numprint{4.86}) & \numprint{582.90} (\numprint{2.60}) & \numprint{648.29} (\numprint{2.90}) & \numprint{223.88} (\numprint{1.00})  \\
        ER $n=2^{16}, m = n \cdot 2^{9}$ & \numprint{106605.57} (\numprint{451.90}) & \numprint{1120.17} (\numprint{4.75}) & \numprint{649.76} (\numprint{2.75}) & \numprint{690.57} (\numprint{2.93}) & \numprint{235.91} (\numprint{1.00})  \\
        ER $n=2^{16}, m = n \cdot 2^{10}$ & \numprint{237704.22} (\numprint{926.70}) & \numprint{1097.79} (\numprint{4.28}) & \numprint{702.56} (\numprint{2.74}) & \numprint{744.72} (\numprint{2.90}) & \numprint{256.51} (\numprint{1.00})  \\
        ER $n=2^{17}, m = n \cdot 2^{8}$ & \numprint{51162.74} (\numprint{162.03}) & \numprint{1280.94} (\numprint{4.06}) & \numprint{671.22} (\numprint{2.13}) & \numprint{849.03} (\numprint{2.69}) & \numprint{315.75} (\numprint{1.00})  \\
        ER $n=2^{18}, m = n \cdot 2^{8}$ & \numprint{54897.60} (\numprint{132.92}) & \numprint{1529.76} (\numprint{3.70}) & \numprint{765.26} (\numprint{1.85}) & \numprint{1022.91} (\numprint{2.48}) & \numprint{413.00} (\numprint{1.00})  \\
        RHG $n=2^{13}, m = n \cdot 2^{8}$ & \numprint{20293.78} (\numprint{141.88}) & \numprint{567.05} (\numprint{3.96}) & \numprint{301.04} (\numprint{2.10}) & \numprint{314.12} (\numprint{2.20}) & \numprint{143.03} (\numprint{1.00})  \\
        RHG $n=2^{14}, m = n \cdot 2^{8}$ & \numprint{22767.57} (\numprint{129.68}) & \numprint{680.43} (\numprint{3.88}) & \numprint{374.21} (\numprint{2.13}) & \numprint{390.13} (\numprint{2.22}) & \numprint{175.57} (\numprint{1.00})  \\
        RHG $n=2^{15}, m = n \cdot 2^{8}$ & \numprint{24662.86} (\numprint{113.24}) & \numprint{838.95} (\numprint{3.85}) & \numprint{440.53} (\numprint{2.02}) & \numprint{480.77} (\numprint{2.21}) & \numprint{217.79} (\numprint{1.00})  \\
        RHG $n=2^{16}, m = n \cdot 2^{6}$ & \numprint{10728.51} (\numprint{51.36}) & \numprint{858.08} (\numprint{4.11}) & \numprint{443.96} (\numprint{2.13}) & \numprint{484.52} (\numprint{2.32}) & \numprint{208.88} (\numprint{1.00})  \\
        RHG $n=2^{16}, m = n \cdot 2^{7}$ & \numprint{18692.50} (\numprint{79.10}) & \numprint{938.83} (\numprint{3.97}) & \numprint{489.00} (\numprint{2.07}) & \numprint{529.89} (\numprint{2.24}) & \numprint{236.32} (\numprint{1.00})  \\
        RHG $n=2^{16}, m = n \cdot 2^{8}$ & \numprint{28621.40} (\numprint{112.17}) & \numprint{966.38} (\numprint{3.79}) & \numprint{496.04} (\numprint{1.94}) & \numprint{558.27} (\numprint{2.19}) & \numprint{255.15} (\numprint{1.00})  \\
        RHG $n=2^{16}, m = n \cdot 2^{9}$ & \numprint{49899.54} (\numprint{169.68}) & \numprint{1054.92} (\numprint{3.59}) & \numprint{568.25} (\numprint{1.93}) & \numprint{621.73} (\numprint{2.11}) & \numprint{294.09} (\numprint{1.00})  \\
        RHG $n=2^{16}, m = n \cdot 2^{10}$ & \numprint{96823.97} (\numprint{291.86}) & \numprint{1165.62} (\numprint{3.51}) & \numprint{634.83} (\numprint{1.91}) & \numprint{685.68} (\numprint{2.07}) & \numprint{331.74} (\numprint{1.00})  \\
        RHG $n=2^{17}, m = n \cdot 2^{8}$ & \numprint{33715.72} (\numprint{100.35}) & \numprint{1225.00} (\numprint{3.65}) & \numprint{608.60} (\numprint{1.81}) & \numprint{705.11} (\numprint{2.10}) & \numprint{335.98} (\numprint{1.00})  \\
        RHG $n=2^{18}, m = n \cdot 2^{8}$ & \numprint{47052.31} (\numprint{113.11}) & \numprint{1428.94} (\numprint{3.43}) & \numprint{691.69} (\numprint{1.66}) & \numprint{829.97} (\numprint{2.00}) & \numprint{416.00} (\numprint{1.00})  \\

\end{tabular}
\caption{Average time in $ns$ per operation (and slowdown to fastest), Section~\ref{exp:m-random}, for different Erd\"os-Renyi and random hyperbolic geometric graph families }
\end{table*}

\begin{table*}[!h]
    \centering
    \begin{tabular}{l|r|r|r|r|r}
$\eta$ & \matchbbhss & \matchbgs & \matchrr & \matchsolomon & \matchtriv \\ \hline \hline
0 & \numprint{1053.71} (\numprint{174.69}) & \numprint{25.33} (\numprint{4.20}) & \numprint{13.07} (\numprint{2.17}) & \numprint{25.29} (\numprint{4.19}) & \numprint{6.03} (\numprint{1.00})  \\
0.25 & \numprint{1057.44} (\numprint{80.77}) & \numprint{25.40} (\numprint{1.94}) & \numprint{13.09} (\numprint{1.00}) & \numprint{24.03} (\numprint{1.84}) & \numprint{14.05} (\numprint{1.07})  \\
0.5 & \numprint{1063.82} (\numprint{80.23}) & \numprint{25.59} (\numprint{1.93}) & \numprint{13.26} (\numprint{1.00}) & \numprint{22.40} (\numprint{1.69}) & \numprint{22.17} (\numprint{1.67})  \\
0.75 & \numprint{1093.46} (\numprint{79.42}) & \numprint{26.28} (\numprint{1.91}) & \numprint{13.77} (\numprint{1.00}) & \numprint{20.71} (\numprint{1.50}) & \numprint{30.55} (\numprint{2.22})  \\
1 & \numprint{1133.05} (\numprint{77.56}) & \numprint{27.08} (\numprint{1.85}) & \numprint{14.61} (\numprint{1.00}) & \numprint{18.97} (\numprint{1.30}) & \numprint{39.32} (\numprint{2.69})  \\
\end{tabular}
\caption{Average time in $ns$ per operation (and slowdown to fastest), Section~\ref{exp:m-windows}}
\end{table*}

\begin{table*}[!h]
    \centering
    \begin{tabular}{l|r|r|r|r|r}
        Graph & \matchbbhss & \matchbgs & \matchrr & \matchsolomon & \matchtriv \\ \hline \hline

dewiki & \numprint{33801.67} (\numprint{239.28}) & \numprint{499.93} (\numprint{3.54}) & \numprint{328.57} (\numprint{2.33}) & \numprint{334.01} (\numprint{2.36}) & \numprint{141.26} (\numprint{1.00})  \\
frwiki & \numprint{71420.54} (\numprint{515.89}) & \numprint{547.24} (\numprint{3.95}) & \numprint{320.29} (\numprint{2.31}) & \numprint{391.16} (\numprint{2.83}) & \numprint{138.44} (\numprint{1.00})  \\
itwiki & \numprint{73972.12} (\numprint{476.47}) & \numprint{530.66} (\numprint{3.42}) & \numprint{337.78} (\numprint{2.18}) & \numprint{398.15} (\numprint{2.56}) & \numprint{155.25} (\numprint{1.00})  \\
nlwiki & \numprint{21734.20} (\numprint{146.27}) & \numprint{516.84} (\numprint{3.48}) & \numprint{321.31} (\numprint{2.16}) & \numprint{398.45} (\numprint{2.68}) & \numprint{148.59} (\numprint{1.00})  \\
plwiki & \numprint{50364.98} (\numprint{349.47}) & \numprint{489.49} (\numprint{3.40}) & \numprint{303.27} (\numprint{2.10}) & \numprint{349.70} (\numprint{2.43}) & \numprint{144.12} (\numprint{1.00})  \\
simplewiki & \numprint{6025.97} (\numprint{70.93}) & \numprint{316.30} (\numprint{3.72}) & \numprint{197.13} (\numprint{2.32}) & \numprint{288.12} (\numprint{3.39}) & \numprint{84.96} (\numprint{1.00})  \\

\end{tabular}
\caption{Average time in $ns$ per operation (and slowdown to fastest), Section~\ref{exp:m-wiki}}
\end{table*}

\end{appendices}

\end{document}